%% file: JASA-template.tex
\documentclass[12pt]{article}
\usepackage{amsmath}
\usepackage{graphicx,psfrag,epsf}
\usepackage{enumerate}
\usepackage{natbib}
\usepackage{url} % not crucial - just used below for the URL 

\usepackage{newtxtext, newtxmath} % Or use 'newpxtext, newpxmath' for Palatino
\usepackage[T1]{fontenc}
\usepackage{microtype}

	\usepackage[english]{babel} 
	\usepackage{authblk}
    
	\usepackage{ textcomp, amsthm, mathtools,mathrsfs} 
	\allowdisplaybreaks
	\usepackage{wasysym} 
	\usepackage{stmaryrd}
	\usepackage{dsfont}
	\usepackage{hyperref}
 \usepackage{fancyhdr}
   	\usepackage{booktabs}
\usepackage{multirow}
    \usepackage{enumerate}
   	\usepackage{inputenc}
	\usepackage{csquotes}
	\usepackage{floatrow}
	\usepackage{footnote,footmisc} 
	\usepackage{threeparttable}
	\usepackage{tikz}
 \usetikzlibrary{calc, decorations.pathreplacing}
\usepackage{comment}
\usepackage{algorithm}
\usepackage{algorithmic}
 \usepackage{subcaption}
 \usepackage{pgfplots}
\usepgfplotslibrary{statistics}
\pgfplotsset{compat=1.16}
\usetikzlibrary{shapes.misc}
	\usetikzlibrary{patterns}
	\usepackage{xcolor}
	\hypersetup{
    	colorlinks,
    	linkcolor={red!50!black},
    	citecolor={blue!50!black},
    	urlcolor={blue!80!black}
	}
	\usetikzlibrary{decorations.pathreplacing}
	\usetikzlibrary{arrows}
	\usepackage{float} 
	\usepackage{pifont}% http://ctan.org/pkg/pifont

\newcommand*\dd{\mathop{}\!\mathrm{d}}

	\usepackage[toc,page]{appendix}
	
	\newtheorem{theorem}{Theorem}[section]
	\newtheorem{lemma}[theorem]{Lemma} 
	\newtheorem{proposition}[theorem]{Proposition}

    \newtheorem{corollary}[theorem]{Corollary}
	
	\newtheorem{example}[theorem]{Example}
	
	\newtheorem*{example*}{Example}

\newtheorem{assumpL}{}

\newtheorem{assumpK}{}

\newtheorem{assumpT}{}

\newenvironment{supplement}{%
    \clearpage
    \thispagestyle{empty}
    \begin{center}
      \huge\bfseries Supplement 
    \end{center}
    \par\bigskip\par
    
    \setcounter{section}{0}
    \setcounter{table}{0}
    \setcounter{figure}{0}
    \setcounter{equation}{0}
    
  }
  {}
  
	\usepackage{graphicx}  
	\usepackage{marvosym}

    \DeclareMathOperator*{\argmin}{\arg\!\min}

    \newtheorem*{assumptions*}{\assumptionnumber}
\providecommand{\assumptionnumber}{}

\usepackage{caption}
\usepackage{mwe}

\definecolor{cyan}{cmyk}{1, 0.4, 0, 0} 

\newcommand{\sumn}{\sum_{i=1}^n}
\newcommand{\var}{\operatorname{var}}

\definecolor{mypink}{RGB}{219, 48, 122}

\def\XXint#1#2#3{{\setbox0=\hbox{$#1{#2#3}{\int}$ }
\vcenter{\hbox{$#2#3$ }}\kern-.6\wd0}}

\makeatletter
\allowdisplaybreaks

 \def\thanks#1{\protected@xdef\@thanks{\@thanks
        \protect\footnotetext{#1}}}
\makeatother

\newcommand{\blind}{1}

\begin{document}

\def\spacingset#1{\renewcommand{\baselinestretch}%
{#1}\small\normalsize} \spacingset{1}

%%%%%%%%%%%%%%%%%%%%%%%%%%%%%%%%%%%%%%%%%%%%%%%%%%%%%%%%%%%%%%%%%%%%%%%%%%%%%%

\if1\blind
{
  \title{\Large \bfseries A Test for Jumps in Metric-Space Conditional Means \\[0.5em]  \vspace{4mm}
 \normalsize With Applications to Compositional and Network Data}
  \author{\normalsize David Van Dijcke\thanks{
     \footnotesize I thank Zach Brown, Florian Gunsilius, Andreas Hagemann, Eva Janssens, Harry Kleyer, Mark Polyak, Konstantin Sonin, Kaspar Wuthrich, and participants at the Michigan Econometrics-IO Workshop, and Econometrics, IO, and Labor Lunches for helpful discussion and comments. I gratefully acknowledge support from the Rackham Pre-Doctoral Fellowship at the University of Michigan. This research was supported in part through computational resources and services provided by Advanced Research Computing at the University of Michigan, Ann Arbor. All errors are mine. An R package accompanying this paper is available at \url{https://github.com/Davidvandijcke/frechesTest}.} \\ \vspace{-2mm}
    Department of Economics, University of Michigan} 
    \date{\vspace{-2mm} \normalsize \today}
  \maketitle
} \fi

\if0\blind
{
  \bigskip
  \bigskip
  \bigskip
  \begin{center}
    {\LARGE\bf Title}
\end{center}
  \medskip
} \fi
\vspace{-2em}
\begin{abstract}

Standard methods for detecting discontinuities in conditional means are not applicable to outcomes that are complex, non-Euclidean objects like distributions, networks, or covariance matrices. This article develops a nonparametric test for jumps in conditional means when outcomes lie in a non-Euclidean metric space. Using local Fr\'echet regression, the method estimates a mean path on either side of a candidate cutoff. This extends existing $k$-sample tests to a non-parametric regression setting with metric-space valued outcomes. I establish the asymptotic distribution of the test and its consistency against contiguous alternatives. For this, I derive a central limit theorem for the local estimator of the conditional Fr\'echet variance and a consistent estimator of its asymptotic variance. Simulations confirm nominal size control and robust power in finite samples. Two empirical illustrations demonstrate the method's ability to reveal discontinuities missed by scalar-based tests. I find sharp changes in (i) work-from-home compositions at an income threshold for non-compete enforceability and (ii) national input-output networks following the loss of preferential U.S. trade access. These findings show the value of analyzing regression outcomes in their native metric spaces.
\end{abstract}

\noindent%
\vfill

\newpage
\spacingset{1.45} % DON'T change the spacing!

\input{body.tex}
\end{document}

%% file: body.tex
\section{Introduction}

Testing for a discontinuity in a conditional mean function is a central task in modern econometrics and statistics, forming the basis of the regression discontinuity (RD) design for causal inference \citep{thistlethwaite1960regression}, changepoint detection \citep{page1957problems}, and structural break tests \citep{chow1960tests}. Under the assumption that the path of conditional means evolves smoothly, discontinuities can reveal important insights into the data-generating process, such as the presence of regime changes, treatment effects, or structural breaks. While central to applied work, standard methods are predicated on scalar-valued outcomes. This is a significant limitation, as many critical questions in science involve outcomes that are inherently complex, non-Euclidean objects. For instance, how does a negative trade shock alter the \textit{network structure} of a national economy?  Or how does crossing a critical temperature threshold trigger an abrupt change in the \textit{distribution} of daily rainfall? Answering such questions requires a methodological framework that can detect jumps or breaks in otherwise smoothly evolving distributions, networks, or other metric-space valued data.

A framework for analyzing such data can be built on the Fr\'echet mean, which generalizes the concept of mathematical expectation to general metric spaces \citep{frechet1948elements}. In a regression context, interest lies specifically in the conditional Fr\'echet mean \citep{petersen2019frechet}, which describes the central tendency of a metric-valued outcome $Y$ as a function of a scalar covariate $X$. A direct test on the conditional Fr\'echet means themselves, however, is generally not feasible. Because these means are elements of the object space, which often lacks additive structure---for example, they are themselves networks or distributions---they cannot be directly compared through differencing. This prevents the use of standard test statistics that rely on simple algebraic comparisons. 

This article develops a novel framework that overcomes this challenge. Since conditional Fr\'echet means cannot be differenced, my approach instead detects discontinuities by comparing their associated conditional Fr\'echet variances, which are scalars. This extends the variance-based logic of \citet{dubey2019frechet} from a simple $k$-sample (or group/multi-sample) comparison  \citep{kruskal1952use} to the more complex, nonparametric regression setting where conditional means and variances evolve as arbitrary smooth functions. To operationalize this, I adapt the local Fr\'echet regression estimator of \citet{petersen2019frechet} to estimate conditional Fr\'echet variances on either side of a jump point, as well as by pooling across the jump. The resulting test statistic compares the pooled variance at the jump point to the one-sided variances, providing the first general method for detecting structural breaks in metric-space regression functions. To establish the asymptotic validity and consistency of this test, I derive a new central limit theorem for the local estimator of the conditional Fr\'echet variance, which generalizes the canonical results of \citet{fan1998efficient}. 

In support of the theoretical results, I report simulation exercises for three metric spaces of interest: univariate density functions, graph Laplacians of networks, and covariance matrices. These results indicate that the test has excellent size control in finite samples and quickly converges to the nominal level. Moreover, it exhibits robust power, reliably detecting discontinuities in all three settings even with modest sample sizes and small jump magnitudes. A comparison to a localized adaptation of the $k$-sample test of \citet{dubey2019frechet} to the regression setting illustrates the usefulness of the local regression approach. While the adaptation of the $k$-sample test does achieve nominal coverage under a piecewise-constant signal, it severely overrejects the null under a piecewise-smooth one, unlike the proposed test.  

To demonstrate the test's practical utility, I present two empirical applications. First, I find a significant structural shift in the composition of work-from-home (WFH) arrangements in Washington State at the income threshold where non-compete agreements become legally enforceable. This suggests that WFH was an important bargaining margin during the sample period, which covered the main COVID-19 years (2020-2023). Second, I detect an abrupt change in the structure of national input-output networks for countries that lose preferential trade access under the US Generalized System of Preferences after its reinstatement in 2015. This indicates that this relatively modest program, which has been subject to repeated legislative lapses and reauthorizations, nonetheless has a meaningful impact on the production networks of developing countries reliant on it. In both applications, existing tests for jumps in scalar-valued outcomes, derived from the complex objects, fail to detect statistically significant jumps. Thus, these applications illustrate the ability of the proposed method to uncover nuanced structural changes that are invisible to methods that ignore the metric space topology. 

This article contributes to several strands of the statistics and econometrics literature. First, it advances the field of Fr\'echet regression \citep{petersen2019frechet} by providing a novel tool for inference on local Fr\'echet regression. It extends existing group-comparison and changepoint tests for non-Euclidean objects \citep{dubey2019frechet, dubey2020frechet, chen2017new, szekely2013energy, wang2024nonparametric, muller2024anova, jiang2024two, zhang2025generalized, zhang2025change} to a fully nonparametric regression setting. Compared to the growing body of inferential work in this area \citep{chen2022uniform, bhattacharjee2023single, kurisu2024empirical, van2025regression}, the test proposed here generalizes to a larger class of metric spaces, works in the nonparametric regression setting, or both, and does so under milder regularity conditions, allowing for broad applicability.

Second, this article contributes to the large literatures on changepoint detection \citep{page1957problems}, regression discontinuity design (RDD) \citep{hahn2001identification}, and structural break estimation \citep{chow1960tests}. My contribution is to develop the first unified framework for testing discontinuities in metric-space conditional means, moving beyond functional data to general metric-space valued data such as distributions, networks, or spheres \citep{Berkes2009Detecting}. Recent research has developed testing approaches tailored to specific types of metric spaces like networks \citep{enikeeva2025change, madrid2023change, kei2024change, penaloza2024changepoint, bhattacharjee2020change} or covariance matrices \citep{dornemann2024detecting, ryan2023detecting, cho2025change, avanesov2018change, dette2022estimating, Aue2009Break}. These methods usually assume specific generative models, such as the stochastic block model for networks, and piecewise-constant signals. By contrast, the proposed test is model-agnostic and allows for piecewise-smooth processes, in addition to being applicable to many different metric spaces. 

Finally, though the focus is on testing, this article contributes to the emerging literature on RDD for non-Euclidean outcomes. For the case of distribution-valued outcomes, \citet{van2025regression} provided the first treatment of this problem, developing bias-corrected inference and establishing the complete inferential framework necessary for applied research. For that particular metric space, the linear structure of univariate quantile functions allowed to move beyond just testing for a jump to also providing point estimates and confidence bands for the jump's magnitude. Conversely, the test developed in the current article is more general and works for metric spaces that do not have such linear structure. Subsequently and independently, \citet{kurisu2025regression} proposed a conceptual extension of RDD to general geodesic metric spaces. While their framework encompasses a broad class of metric spaces, it focuses primarily on identification and estimation, leaving the question of inference—and particularly hypothesis testing—unaddressed. The present article contributes to this literature by developing the first formal test for discontinuities in conditional Fr\'echet means in general metric spaces, providing an inferential tool that helps operationalize metric space RDDs for practical applications. In doing so, I build upon the foundation established in \citet{van2025regression} while extending the scope to the general metric spaces considered in \citet{kurisu2025regression}.

The remainder of the article develops the methodology and main theoretical results (Sections \ref{sec:setup_definitions}--\ref{sec:main_results}), presents simulation evidence (Section \ref{sec:simulations}), illustrates the method with two empirical applications (Section \ref{sec:applications}), and concludes.

\section{Setup and Motivation} \label{sec:setup_definitions}

Let $(\Omega, d)$ be a bounded metric space, where $d:\Omega \times \Omega \to [0, \infty)$ is a distance function. Let $Y$ be a random object taking values in $\Omega$, and $X$ be a real-valued random variable representing a covariate or conditioning variable. We are interested in the behavior of $Y$ conditional on $X$. Denote their joint distribution $F$ and the conditional distribution of $Y$ given $X$ as $F_{Y \mid X}$.

\subsection{Conditional Fr\'echet Means and Variances}

The concept of a conditional Fr\'echet mean arises as a natural generalization of the Euclidean mean. Let $Z \in \mathbb{R}$, then the conditional expectation $E[Z \mid X=x]$ can be defined as the unique minimizer $f(x)$ of the mean squared error,
\[
E[Z \mid X=x] \coloneqq \argmin_{f(x) \in \mathbb{R}} E[ d_E\left(Z,f(x)  \right)^2 \mid X=x],
\]
where $d_E(z_1,z_2) \coloneqq \| z_1 - z_2\|$ the standard Euclidean metric. 

The conditional Fr\'echet mean of $Y \in \Omega$ given $X=x$, where $Y$ is a metric-space valued \textit{random} object, is defined analogously by replacing the Euclidean metric with the more general distance metric $d$ \citep{petersen2019frechet},
\begin{equation} \label{eq:cond_frechet_mean_def}
m_{\oplus}(x) \coloneqq \argmin_{\omega \in \Omega} M_{\oplus}(\omega, x), \quad \text{where } M_{\oplus}(\omega, x) \coloneqq E\left[d^2(Y, \omega) \mid X=x\right],
\end{equation}
and note that the expectation is taken over $F_{Y \mid x}$. In words, the conditional Fr\'echet mean $m_{\oplus}(x)$ is an element in $\Omega$ that minimizes the expected squared distance to $Y$ among observations with $X=x$. In that sense, it generalizes the standard Euclidean conditional expectation, which it includes as a special case. Beyond this natural interpretation as a generalized expectation, an important benefit of working with the Fr\'echet mean is that it lies in the metric space itself and hence captures its topology (e.g., it is a network Laplacian), in contrast to other scalar-valued statistics that one may derive from the metric space-valued objects (e.g., the average centrality of the networks). Of course, what aspect of its topology it captures is determined by the choice of distance metric $d$. Nonetheless, the test proposed in this article works for many metrics and hence does not hinge on a specific choice. 

The corresponding conditional Fr\'echet variance is,
\begin{equation} \label{eq:cond_frechet_var_def}
    V_{\oplus}(x) \coloneqq E \left[ d^2(Y, m_{\oplus}(x)) \mid X=x \right] = M_{\oplus}(m_{\oplus}(x), x).
\end{equation}
This measures the expected squared dispersion of $Y$ around its conditional Fr\'echet mean $m_{\oplus}(x)$, analogous to the standard Euclidean variance. 

We are interested in a potential discontinuity in the conditional Fr\'echet mean at a specific point $X=c \in \mathbb{R}$. To formalize this, define the conditional Fr\'echet means from the left and right of $c$,
\begin{equation} \label{eq:frechet_cond_mean_pm}
m_{\pm, \oplus}\coloneqq \underset{\omega \in \Omega}{\operatorname{argmin}} \left\{ \lim_{x \to c^{\pm}} M_{\oplus}(\omega, x) \right\},
\end{equation}
These represent the central tendencies of $Y$ as $X$ approaches $c$ from below ($m_{-, \oplus}$) or above ($m_{+, \oplus}$). The corresponding ``one-sided'' conditional Fr\'echet variances are,
\begin{equation} \label{eq:frechet_variance_population_pm}
    V_{\pm, \oplus} \coloneqq \lim_{x \to c^{\pm}} E \left[ d^2(Y, m_{\oplus}(x)) \mid X=x \right].
\end{equation}
Assuming these limits exist, one can also relate them to underlying ``counterfactual'' Fr\'echet means at $X=c$ for a connection to causal inference, specifically the regression discontinuity design \citep{van2025regression, kurisu2025regression}. Since the focus of this article is on testing, I do not discuss this further here.

The null hypothesis of no discontinuity in the conditional Fr\'echet mean at $c$ is,
\begin{equation*}
H_0^{\text{mean}}: m_{+,\oplus} = m_{-,\oplus},
\end{equation*}
where $Y_1=Y_2$ for $Y_1, Y_2 \in \Omega$ is defined as $d(Y_1,Y_2)=0$.
The null hypothesis of no discontinuity in the conditional Fr\'echet variance is,
\begin{equation*}
H_0^{\text{var}}: V_{+,\oplus} = V_{-,\oplus}.
\end{equation*}
Similar to the classical analysis of variance test, the proposed test will assess $H_0^{\text{mean}}$ and $H_0^{\text{var}}$ jointly, that is,
\begin{equation*} 
H_0 :  H_0^{\text{mean}}  \cap H_0^{\text{var}},
\end{equation*}
while the alternative, $H_1$, is that at least one of these conditions does not hold.

\subsection{Motivating Settings} \label{sec:motivating_setting}

To motivate the problem further, I now discuss several classes of problems and specific applications where detecting jumps in the conditional mean path of a metric-space valued outcome can be highly valuable. Below, I further motivate the test with two empirical illustrations on real-world data. 

\begin{example}[Regression Discontinuity Design (RDD)]
    The RDD is a powerful quasi-experimental design for causal inference \citep{hahn2001identification, thistlethwaite1960regression}, hinging on the estimation of a discontinuity in $E[Z|X=x]$ for $Z \in \mathbb{R}$ at a known threshold $c$ of a running variable $X$. The proposed framework allows the outcome to be a complex, non-Euclidean object.
    \begin{itemize}
         \item \emph{Income Distribution Response to Democratic Governors.} Consider a ``close-election'' RDD which exploits the fact that when a Democratic governor's vote share ($X$) exceeds 50\% of the two-party vote, they become elected ($T$). The outcome $Y_i$ for state $i$ is the distribution of family income in that state. A suitable metric space for this setting is the space of cumulative distribution functions equipped with the 2-Wasserstein distance. This is the application in \citet{van2025regression}, which found negative effects of Democratic governorships on the top of the income distribution but not elsewhere. As mentioned, for this particular metric space, one can also construct tests using the uniform confidence bands developed in that article. 
    \end{itemize}
    The test proposed here allows for formally assessing if $m_{+,\oplus} \neq m_{-,\oplus}$ in such RDD settings. See also \citet{kurisu2025regression} for further examples of metric-space valued RD designs.
\end{example}

\begin{example}[Changepoint Detection in Metric-Space Valued Processes]
    A central task in many disciplines is to identify changepoints where the underlying data-generating process for observations $Y_x$ (indexed by a continuous variable $X$, such as time or a spatial coordinate) alters its central tendency \citep{aue2013break}. A statistically significant jump in $m_{\oplus}(x)$ at $X=c$ indicates that such a  changepoint exists.
    \begin{itemize}
        \item \emph{Neuroscience: Brain Development Trajectories:} Let $X$ be the chronological age of subjects. $Y_x$ could represent the structural brain network of an individual of age $x$, perhaps derived from diffusion tensor imaging (DTI) and viewed as a graph object. Using a graph metric (e.g., based on the Frobenius norm or spectral graph distances), a jump in the Fr\'echet mean brain network structure at a critical developmental age $c$ could highlight a significant maturational shift in the topology of typical brain connectivity \citep{fair2009functional, dubey2019frechet}.
    \end{itemize}

\end{example}

\begin{example}[Threshold Effects in Complex Systems]
    More broadly, the framework applies to detecting threshold effects where the central tendency of a metric-space outcome $Y$ changes abruptly when a continuous underlying factor $X$ crosses a critical value $c$. This generalizes the search for non-linearities and critical transitions.
    \begin{itemize}
        \item \emph{Ecology: Community Composition along Environmental Gradients:} Let $X$ be a continuous environmental gradient (e.g., soil pH, salinity). $Y_x$ is the ecological community composition (e.g., a vector of species abundances, or a set of species presence/absence) at a site with gradient value $x$. Using a suitable distance metric such as the Aitchison or geodesic distance, a jump in the Fr\'echet mean community composition at a specific gradient level $c$ would indicate an ecotone or an abrupt shift in the typical community structure \citep{killick2012optimal}.
    \end{itemize}
\end{example}

In all these settings, the objects $Y_i$ possess a complex structure. The core statistical question is whether the conditional Fr\'echet mean path $m_{\oplus}(x)$, ``$E[Y|X=x]$'',  exhibits a significant jump at $X=c$. The subsequent sections develop the formal tools to test this hypothesis.

\subsection{One-Sided Local Fr\'echet Regression} \label{sec:twosidedmonster}

To estimate the conditional means and variances from the left and right, $m_{\pm, \oplus}, V_{\pm, \oplus}$, I rely on a one-sided version of local Fr\'echet regression \citep{petersen2019frechet}. Assume one observes i.i.d. data $\{(Y_i, X_i)\}_{i=1}^n \sim F$. Then, the one-sided estimators are defined as,
\begin{equation} \label{eq:l_hat_pm}
\hat{l}_{\pm,\oplus} \coloneqq \underset{\omega \in \Omega}\argmin \: \left\{ \hat{L}_{\pm,n}(\omega) \coloneqq  n^{-1} \sum_{i=1}^n s_{\pm, i n}(c, h_m) d^2\left(Y_i, \omega\right) \right\},
\end{equation}
where $h_m$ is a bandwidth and $s_{\pm,in}(c,h)$ are standard one-sided local linear regression weights that use data only to the left or right of the cutoff, respectively,
\begin{align*}
s_{+,in}(c,h) &= \mathrm{1}(X_i \geq c) \frac{K_h(X_i-c)}{\hat{\sigma}_{+,0}^2} \left[ \hat{\mu}_{+,2} - \hat{\mu}_{+,1}(X_i-c) \right], \\
\intertext{where $K_h(u) = K(u/h)/h$ for a kernel function $K$, and the components are sample moments calculated using only data to the right of the cutoff,}
\hat{\mu}_{+,j} &= n^{-1} \sum_{k=1}^n \mathrm{1}(X_k \geq c) K_h(X_k-c)(X_k-c)^j \quad \text{for } j=0,1,2, \\
\hat{\sigma}_{+,0}^2 &= \hat{\mu}_{+,0}\hat{\mu}_{+,2} - \hat{\mu}_{+,1}^2.
\end{align*}
The weights for the left-sided estimator, $s_{-, i n}(c, h)$, are defined analogously using observations where $X_i < c$ and $X_k < c$. 

The corresponding population (or pseudo-true) minimizers are denoted $\tilde{l}_{\pm, \oplus}$ with objective function $\tilde{L}_{\pm, \oplus}(\omega)$. The consistency of these estimators is a straightforward extension of the arguments in \citet{petersen2019frechet} and \citet{kurisu2025regression}.

The sample estimators of the conditional Fr\'echet variances from above ($V_{+,\oplus}$) and below ($V_{-,\oplus}$) are then simply,
\begin{equation} \label{eq:v_hat_split_main}
\hat{V}_{\pm, \oplus} = \frac{1}{n} \sum_{i=1}^n s_{\pm,in}(c,h) d^2(Y_i, \hat{l}_{\pm, \oplus}),
\end{equation}
where I use a different bandwidth $h$ for the variance than for the mean estimation ($h_m$) to allow for different degrees of smoothing. The asymptotic variance of $d^2(Y, m_{\pm,\oplus})$ conditional on $X=c$ is denoted $\sigma_{\pm, V}^2 \coloneqq \operatorname{var}\left(d^2(Y, m_{\pm, \oplus}) \mid X=c \right)$, and its sample estimator $\hat{\sigma}_{\pm, V}^2$ is constructed as,
\[
\hat{\sigma}_{\pm, V}^2=\frac{1}{n} \sumn s_{\pm,in}(c,h)  d^4(Y_i, \hat{l}_{\pm, \oplus}) - \left(\frac{1}{n} \sumn s_{\pm,in}(c,h)  d^2(Y_i, \hat{l}_{\pm, \oplus})\right)^2,
\]
by the variance shortcut formula. 
The asymptotic normality of $\hat{V}_{\pm, \oplus}$ and consistency of $\hat{\sigma}_{\pm, V}^2$ are established below. Bandwidth selection is done by a cross-validation procedure that aims to minimize the out-of-sample prediction on both sides of the cutoff simultaneously (see Appendix \ref{app:bandwidth} for more details).

\subsection{Pooled Estimators} \label{sec:pooled_estimators}
To construct the ANOVA-style test, I also need pooled estimators that do not distinguish between observations above and below $c$. The pooled local Fr\'echet mean estimator is,
\[
\hat{l}_{p,\oplus} \coloneqq \argmin_{\omega \in \Omega} \frac1n \sumn s_{p,in}(c,h) d^2(Y_i, \omega),
\]
and the pooled conditional Fr\'echet variance estimator is,
\[
\hat{V}_{p,\oplus} \coloneqq \frac1n \sumn s_{p,in}(c,h) d^2(Y_i, \hat{l}_{p,\oplus}),
\]
where $s_{p,in}(c,h) = \frac12 s_{+,in}(c,h) + \frac12 s_{-,in}(c,h)$. The use of uniform weights $\frac12$ assumes that the same bandwidth $h$ is used on both sides, which can easily be relaxed. The corresponding pooled population moments are,
\begin{align*}
m_{p,\oplus} &  \coloneqq \argmin_{\omega \in \Omega} \left\{ M_{p,\oplus}(\omega) \coloneqq \frac12 M_{+,\oplus}(\omega) + \frac12 M_{-,\oplus}(\omega) \right\} \\ 
V_{p,\oplus} & \coloneqq \frac12 M_{+,\oplus}(m_{p,\oplus}) + \frac12 M_{-,\oplus}(m_{p,\oplus}),
\end{align*}
i.e., the pooled population estimators pool across the one-sided objective functions. 
The consistency of these pooled estimators is proved in Proposition \ref{prop:pooled_consistency} in the Appendix.

\subsection{A Jump Test for Conditional Metric-Space Means}

To test for a discontinuity, I propose an ANOVA-style test statistic based on these one-sided and pooled local Fr\'echet regression estimators. Define auxiliary statistics,
\begin{align}
F_n &=  \hat{V}_{p, \oplus} - \left( \frac12 \hat{V}_{+, \oplus} + \frac12 \hat{V}_{-, \oplus}  \right), \label{eq:F_n} \\
U_n &= \frac{\left(\hat{V}_{+, \oplus}-\hat{V}_{-, \oplus}\right)^2}{ (S_+ \hat{\sigma}_{+, V}^2 / \hat{f}_X(c)) + (S_- \hat{\sigma}_{-, V}^2 / \hat{f}_X(c)) }. \label{eq:U_n}
\end{align}
Here, $U_n$ tests for differences in variances ($H_0^{\text{var}}$). Under the null of equal variances (i.e. $U_n \to 0$), any difference between the pooled and one-sided variances reflected in $F_n$ must be driven by differences in means. Hence, $F_n$ tests for differences in means ($H_0^{\text{mean}}$).

The combined test statistic is,
\begin{equation} \label{eq:test_combined_main}
T_n=nh_n U_n +\frac{nh_n F_n^2}{ (\hat{\sigma}_{+,V}^2 S_+/\hat{f}_X(c)) + (\hat{\sigma}_{-,V}^2 S_-/\hat{f}_X(c)) }.
\end{equation}
Here, the $F_n$ term is squared and scaled by its variance estimate so that it has the same variability as the $U_n$ term. 

\section{Main Results} \label{sec:main_results}

\subsection{Assumptions} \label{sec:assumptions}
I require the following assumptions for the asymptotic results. Throughout, I assume that the metric space $\Omega$ is bounded.

\begin{assumpL}[Sampling] \label{asspt:sampling}
$\{ (Y_i, X_i) \}_{i=1}^n$ are i.i.d. copies of a random element $(Y,X)$ defined on a probability space $(\Xi, \mathcal{F}, P)$. $Y \in \Omega$, $X \in \mathbb{R}$. 
\end{assumpL}
\begin{assumpL}[Densities and Continuity]   \label{asspt:densities}
The following hold for some $\varepsilon>0$. (a) The density $f_X$ is twice continuously differentiable on
      $(c-\varepsilon,c+\varepsilon)$, bounded away from $0$ and $\infty$,
      and $f_X(c)>0$. (b) For every $y\in\Omega$ the function
      $g_y(x)=\partial_x F_{X\mid Y}(x\mid y)$ is twice continuously
      differentiable on $(c-\varepsilon,c+\varepsilon)$, and  
      $
         \sup_{y\in\Omega}\;
         \max_{k=0,1,2}\;
         \sup_{x\in(c-\varepsilon,c+\varepsilon)}
         |g_y^{\,(k)}(x)|\;<\;\infty.
      $
      (c) For every open set \(U\subset\Omega\) with
      $F_Y(\partial U)=0$ the map
      $
         x \;\longmapsto\;
         \int_U dF_{Y\mid X=x}(y)
      $
      is twice continuously differentiable on the open
      intervals $(c-\varepsilon,c)$ and $(c,c+\varepsilon)$.  
\end{assumpL}
For each side separately, this assumption is identical to the one in \citet{petersen2019frechet}, and generalizes the standard distributional assumptions imposed in local polynomial regression.

\begin{assumpL}[Uniqueness of Fr\'echet Means] \label{asspt:uniqueness}
The Fr\'echet means $m_{\pm,\oplus}$ and $m_{p,\oplus}$ exist and are unique. For all $n$, their estimators $\tilde{l}_{\pm,\oplus}$, $\hat{l}_{\pm,\oplus}$, $\tilde{l}_{p,\oplus}$, $\hat{l}_{p,\oplus}$ exist and are unique (the hatted estimators almost surely). Additionally, for any $\varepsilon>0$: $
\inf _{d\left(\omega, m_{\pm,\oplus}\right)>\varepsilon}\left\{M_{\pm,\oplus}(\omega)-M_{\pm,\oplus}\left(m_{\pm,\oplus}\right)\right\}>0,$ 
$\liminf_n \inf _{d\left(\omega, \tilde{l}_{\pm,\oplus}\right)>\varepsilon}\left\{\tilde{L}_{\pm,n}(\omega)-\tilde{L}_{\pm,n}\left(\tilde{l}_{\pm,\oplus}\right)\right\}>0$,  
$ \liminf _n \inf _{d\left(\omega, \hat{l}_{\pm,\oplus}\right)>\varepsilon}\left\{\hat{L}_{\pm,n}(\omega)-\hat{L}_{\pm,n}\left(\hat{l}_{\pm,\oplus}\right)\right\}>0  \quad \text{(a.s.)}$.
Analogous conditions hold for the pooled estimators $\tilde{l}_{p,\oplus}$ and $\hat{l}_{p,\oplus}$ with respect to $M_{p,\oplus}(\omega)$, denoted as Assumption \ref{asspt:uniqueness}-p.
\end{assumpL}

\begin{assumpL}[Curvature Conditions] \label{asspt:curvature} 
\begin{enumerate}[(i)] 
\item \label{asspt:curvature_population}
Provided $d\left(\omega, m_{\pm,\oplus}\right)<\eta_1$, there exist $\eta_1>0, C_1>0$ and $\beta_1>1$ such that,
$$
M_{\pm,\oplus}(\omega)-M_{\pm,\oplus}\left(m_{\pm,\oplus}\right) \geq C_1 d\left(\omega, m_{\pm,\oplus}\right)^{\beta_1}.
$$
\item \label{asspt:curvature_estimator}
Provided $d\left(\omega, \tilde{l}_{\pm,\oplus}\right)<\eta_2$, there exists $\eta_2>0, C_2>0$ and $\beta_2>1$ such that:
$$
\underset{n}{\liminf }\left[\tilde{L}_{\pm,n}(\omega)-\tilde{L}_{\pm,n}\left(\tilde{l}_{\pm,\oplus}\right)\right] \geq C_2 d\left(\omega, \tilde{l}_{\pm,\oplus}\right)^{\beta_2}.
$$
\end{enumerate}
Analogous conditions hold for $M_{p,\oplus}(\omega)$ and $\tilde{L}_{p,n}(\omega)$, denote these as Assumption \ref{asspt:curvature}-p.
\end{assumpL} 

As mentioned, the uniqueness and curvature assumptions are satisfied for all specific metric spaces considered in this article \citep{petersen2019frechet, zhou2022network}. More generally, \citet{KimuraBondell2025Frechet} characterize a class of metric spaces that have unique conditional Fr\'echet means. 

The following are generalizations of standard local polynomial regression assumptions,
\begin{assumpK}[Kernel] \label{asspt:kernel}
The kernel $K: \mathbb{R} \to [0,\infty)$ is a continuous probability density function, symmetric around zero, with compact support.
\end{assumpK}
\begin{assumpK}[Bandwidth] \label{asspt:bandwidth}
(a) The bandwidth sequence $h=h(n)$, satisfies $h \to 0, \, h_m \to 0$, $nh \to \infty, \, n h_m \to \infty$. 
(b) $nh h_m^{4/(\beta_1-1)} \to 0$.
(c) $h_m(n) = \varrho_m n^{-\theta}$ for some $\varrho_m > 0$, $\theta > 0$, and $h(n) = \varrho n^{-\gamma}$ for some $\varrho > 0$, $\gamma > \theta$, $\gamma > \frac15$.
(d) $(nh)^{1/2} (nh_m)^{-1/(2(\beta_2-1))} \to 0$, where $\beta_1, \beta_2$ control the curvature of the Fr\'echet mean estimators in Assumption \ref{asspt:curvature}. 
\end{assumpK}
First, the variance bandwidth $h$ must undersmooth the variance estimator $(\gamma > 1/5$) in (c) to eliminate its own asymptotic bias. Second, the variance bandwidth $h$ is assumed to converge faster than the mean bandwidth $h_m$ ($\gamma > \theta$). This baseline relationship ensures that for sufficiently regular problems ($1 < \beta_2 \le 2$), the stochastic error from the first stage is automatically controlled. The relative rate condition in (b) which handles the bias component of the first-stage error and essentially requires that the mean bandwidth $h_m$ cannot converge too slowly. This ensures the mean estimator $\hat{l}_{\oplus}$ does not itself need to be undersmoothed---an adaptive property that generalizes the findings of \citet{fan1998efficient}. Finally, the condition in (d) handles the stochastic component for less regular problems ($\beta_2 > 2$), demanding a more aggressive separation between the bandwidths as the Fr\'echet mean's objective function becomes less smooth. Note that $\beta_2 \leq 2$ for all metric spaces considered in this article, see \citet{petersen2019frechet, zhou2022network}. These explicit conditions are necessary because, unlike in the Euclidean case, there are no algebraic cancellations to automatically suppress the first-stage error in a general metric space.

Finally, let $N(\epsilon, S, d)$ be the $\epsilon$-covering number of a set $S \subseteq \Omega$. The entropy integral for a $\delta$-ball $B_\delta(\omega_0)$ around $\omega_0 \in \Omega$ is $J(\delta, \omega_0) \coloneqq \int_0^1 \sqrt{1+\log N\left(\delta \epsilon, B_\delta(\omega_0), d\right)} \mathrm{d} \epsilon $. 

\begin{assumpT}[Local Entropy Condition] \label{assptT:J}
For any $\omega_0 \in \{m_{-,\oplus}, m_{+,\oplus}, m_{p,\oplus}\}$, $\delta J(\delta, \omega_0) \rightarrow 0$ as $\delta \rightarrow 0$.
\end{assumpT}

\begin{assumpT}[Global Entropy Condition] \label{asspT:entropy}
The entropy integral of $\Omega$ is finite: $\int_0^1\sqrt{1+\log N(\epsilon, \Omega, d)} \mathrm{d} \epsilon < \infty$.
This is a stronger condition, used for uniform consistency results (e.g., Theorem \ref{thm:consistency}).
\end{assumpT}
 Note that for most results in this article, I only need the weaker local entropy condition \ref{assptT:J}, which is also weaker than its equivalent in \citet{dubey2019frechet} since the latter imposes it for all $\omega \in \Omega$. The stronger global entropy condition \ref{asspT:entropy} is only required for the power results.

\subsection{Central Limit Theorem for Conditional Fr\'echet Variance}
First, I establish the asymptotic distribution of the estimated conditional Fr\'echet variances $\hat{V}_{\pm, \oplus}$.
 Denote convergence in distribution with $\xrightarrow[]{d}$. 

\begin{theorem}[CLT for Conditional Fr\'echet Variance] \label{thm:CLT}
Let Assumptions \ref{asspt:sampling}-\ref{assptT:J} hold. Let $\hat{f}_X(c)$ be a $\sqrt{nh}$-consistent estimator for $f_X(c)$.
Then,
\[
\sqrt{nh} \left( \hat{V}_{\pm, \oplus} - V_{\pm, \oplus}  \right) / \left( \sqrt{\frac{S_\pm}{\hat{f}_X(c)}} \hat{\sigma}_{\pm, V} \right)    \xrightarrow[]{d} N\left(0, 1\right),
\]
where \[ 
S_\pm \coloneqq \frac{\int_{0}^\infty (K_{\pm,12} - u K_{\pm,11})^2 K^2(u) \dd u}{ (K_{\pm,12} K_{\pm,10} - K^2_{\pm,11})^2 }. \]
\end{theorem}
The proof, detailed in Appendix \ref{app:anova_proofs}, follows by decomposing the estimation error into a Bahadur-style representation. The leading stochastic term is shown to be asymptotically normal via a Lyapunov CLT, while the plug-in and bias terms are asymptotically negligible under our bandwidth assumptions. The key intuition behind why one can establish a classical CLT is that all the geometry of the complex objects is absorbed into the squared distance
$d^{2}(Y_i,m_{\pm,\oplus})$. Further, because the metric space
\((\Omega,d)\) is assumed to be bounded, those distances possess uniformly bounded
moments of every order, delivering the necessary control over the asymptotic distribution.

\subsection{Asymptotic Distribution of the Test}

The following results determine the asymptotics of the test statistics in Eqs. \eqref{eq:F_n} and \eqref{eq:U_n} under the respective null hypotheses,
\begin{proposition} \label{prop:F_n}
 Under $H_0^{\text{mean}}: m_{+,\oplus}=m_{-,\oplus}$ and Assumptions \ref{asspt:sampling}-\ref{assptT:J} (including \ref{asspt:uniqueness}-p, \ref{asspt:curvature}-p for pooled estimators),
$$
 \sqrt{nh_n} F_n = o_p(1).
$$
\end{proposition}

\begin{proposition} \label{prop:U_n}
 Under $H_0^{\text{var}}: V_{+,\oplus}=V_{-,\oplus}$, and Assumptions \ref{asspt:sampling}-\ref{assptT:J},
$$
nh_n U_n \xrightarrow[]{d} \chi_{1}^2,
$$
where $\chi_{1}^2$ is the chi-squared distribution with 1 degree of freedom.
\end{proposition}
Proofs are in Appendix \ref{app:anova_proofs}. The result for $U_n$ follows straightforwardly from the central limit theorem derived in the previous section. The result for $F_n$ is intuitive: under the null, the variance estimators converge to the same quantity, and their difference converges faster to $0$ than each separate term does.

 For the combined test statistic in \eqref{eq:test_combined_main}, the following result obtains.
\begin{corollary} \label{cor:T_n_main}
Under $H_0: m_{+,\oplus}=m_{-,\oplus}$ and $V_{+,\oplus}=V_{-,\oplus}$, and Assumptions \ref{asspt:sampling}-\ref{assptT:J} (including -p versions), 
$$
T_n \to_d \chi_{1}^2.
$$
\end{corollary}
This follows because, under the joint null, the contribution of the $F_n$ term vanishes and hence the asymptotic distribution of the $U_n$ term dominates.  Based on this corollary, the rejection region for a test of size $\alpha$ is $R_{n, \alpha}=\{T_n>\chi_{1, \alpha}^2\}$. Since this is an asymptotic result, one may worry that the $F_n$ term can lead to size distortions in finite sample. While the simulations in Section \ref{sec:simulations} do reveal slight overcoverage in very small samples, the test's acceptance probability rapidly converges to the nominal level.

\subsection{Consistency of the Test} \label{sec:consistency}
To study the power of the test, define the population counterparts,
\[
F_{pop}=V_{p,\oplus} - \left( \frac12 V_{+,\oplus} +  \frac12 V_{-,\oplus} \right), \quad U_{pop}=\frac{\left(V_{+, \oplus}-V_{-, \oplus}\right)^2}{\left(S_{+} \sigma_{+, V}^2/f_X(c)\right) + \left(S_{-} \sigma_{-, V}^2/f_X(c)\right)}.
\]

\begin{proposition} \label{prop:F_n_F_consistency}
Under Assumptions \ref{asspt:sampling}-\ref{assptT:J} (including -p versions) and \ref{asspT:entropy}, as $n \to \infty$,
$|F_n - F_{pop}| = o_p(1)$. $F_{pop} \ge 0$, and $F_{pop}=0$ if and only if $m_{p,\oplus} = m_{+,\oplus} = m_{-,\oplus}$.
Also, $|U_n - U_{pop}| = o_p(1)$, and $U_{pop} \ge 0$, with $U_{pop}=0$ if and only if $V_{+,\oplus}=V_{-,\oplus}$.
\end{proposition}

This proposition provides the crucial link between the test statistics and the hypotheses, showing that $F_n$ and $U_n$ are consistent for population measures of mean and variance differences, respectively. These population measures are zero if and only if no discontinuity exists. This equivalence is immediate for $U_n$, while for $F_n$, it follows from the uniqueness of the Fr\'echet means.

Finally, consider sequences of alternatives $H_A^{(n)}$ where $F_{pop} \ge a_n$ or $U_{pop} \ge b_n$ for sequences $a_n, b_n \ge 0$ where either $a_n >0$ or $b_n>0$. The power function is $\beta_{H_A^{(n)}}=\inf_{(F_{pop}, U_{pop}) \in H_A^{(n)}} P(T_n > \chi_{1,\alpha}^2)$.

\begin{theorem}[Consistency of the Test] \label{thm:consistency}
Under Assumptions \ref{asspt:sampling}-\ref{asspT:entropy} (including -p versions), for sequences of alternatives $H_A^{(n)}$ and any $\alpha >0$,
\begin{enumerate}[(a)]
\item If $F_{pop} \ge a_n$ with $(nh_n)^{1/2} a_n \rightarrow \infty$, then $\beta_{H_A^{(n)}} \rightarrow 1$.
\item If $U_{pop} \ge b_n$ with $nh_n b_n \rightarrow \infty$, then $\beta_{H_A^{(n)}} \rightarrow 1$.
\end{enumerate}
\end{theorem}
 This shows the test will pick up any jump whose size does not shrink faster than the sampling noise inside the bandwidth window, i.e., the test is consistent against contiguous alternatives.

\section{Simulations}
\label{sec:simulations}

To evaluate the test's finite-sample performance, I carry out Monte Carlo simulations in three metric spaces: probability densities, covariance matrices, and graph Laplacians of networks. I consider both piecewise-smooth and piecewise-constant data generating processes (DGPs) to assess the usefulness of the nonparametric regression approach (see Appendix~\ref{app:simulation} for full details). Though there are no general metric-space tests that directly compete with the proposed one, I adapt the $k$-sample test of \citet{dubey2019frechet} to compare the two means in a shrinking window around the cutoff to illustrate the value of the nonparametric regression approach. Since this is essentially equivalent to a local constant regression approach, we should expect it to do well on the piecewise-constant signal but overreject on the piecewise-smooth one.

\begin{table}[ht!] 
  \centering      
  \footnotesize \input{tabs/simulation_results_tabular} \caption{Size and Power by Sample Size and Metric Space}
  \label{tab:sim_results}
  \floatfoot{\scriptsize \textit{Note:}
Empirical rejection rates of the proposed Fr\'echet jump test from 1000 Monte Carlo simulations at a nominal $\alpha=0.05$ level. `Size` is the rejection rate under the null hypothesis of no jump ($H_0$). `Power` is the rejection rate under a fixed alternative ($H_1$) with a jump introduced at $c=0.5$. The Data Generating Process is piecewise-smooth in all cases. The fixed jump magnitudes used to calculate `Power` are: $\delta_D=1.5$ for Densities, $\beta_C=1.5$ for Covariance matrices, and $\delta_N=0.25$ for Networks.}
\end{table}

The results confirm the asymptotic theory. Table~\ref{tab:sim_results} shows that the proposed test exhibits excellent size control, with empirical rejection rates under the null quickly converging to the nominal 5\% level as the sample size increases. The test also demonstrates robust power across all settings. Figure~\ref{fig:all_power_curves_combined_stacked} highlights the key advantage of the regression-based approach. While both tests perform well under a simple piecewise-constant signal (right panels) -- with the $k$-sample test expectedly converging faster due to its parametric rate --, the localized $k$-sample test severely over-rejects the null under the piecewise-smooth DGPs (left panels). My proposed test, by contrast, maintains correct size while reliably detecting true jumps, confirming its utility for a broader and more realistic class of problems.

 \begin{figure}[ht!]
    \centering % Center the figure block
    % Replace 'power_curves_stacked_N200.png' with the actual filename generated by your R script
\includegraphics[width=0.99\textwidth]{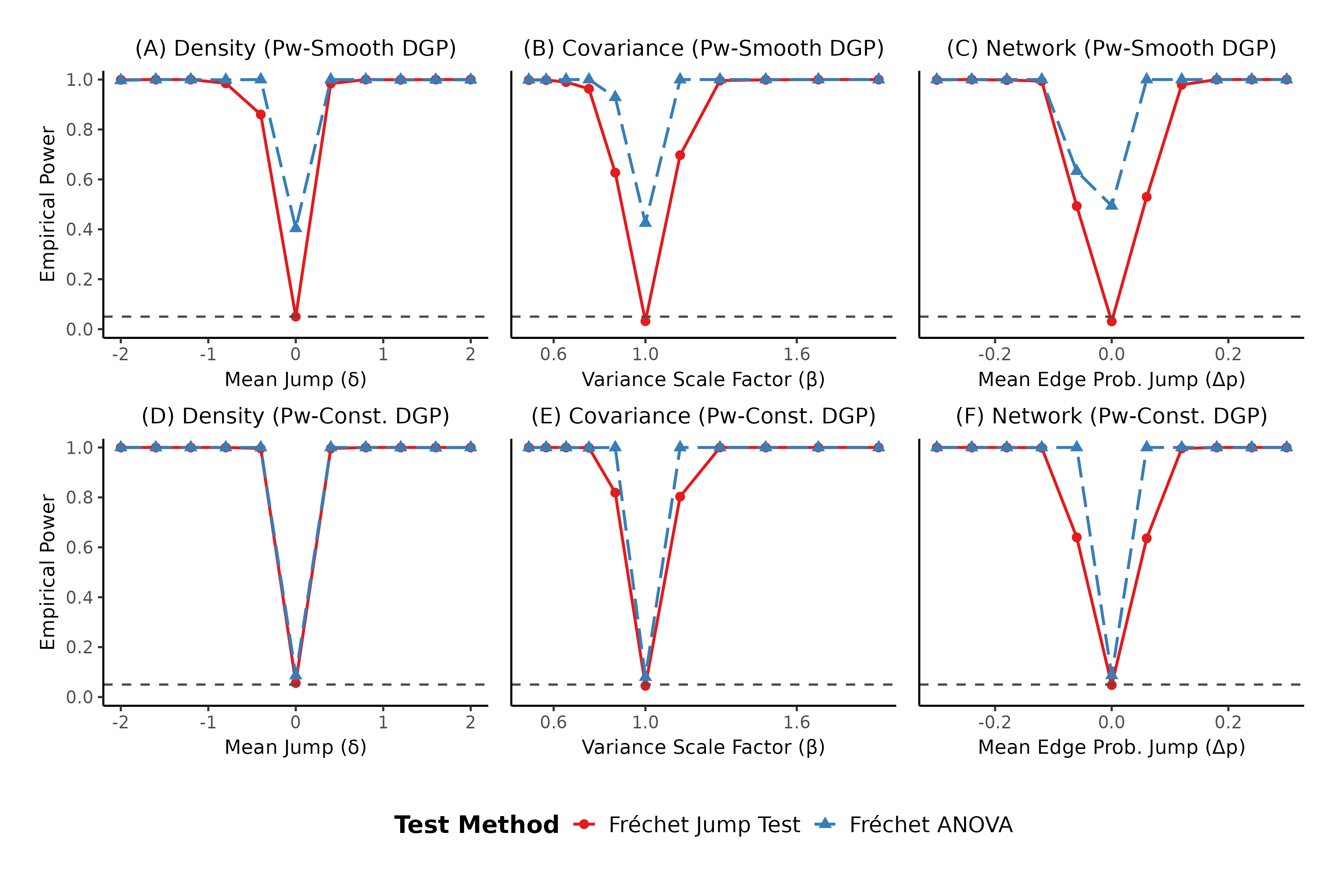}
\vspace{-6mm}
    \caption{Power Curves}
    \floatfoot{\scriptsize \textit{Note:} Figure shows empirical power curves comparing the proposed Fr\'echet jump test and a localized adaptation of the two-sample test from \citet{dubey2019frechet}. x-axis represents the jump magnitude: additive mean jump for densities ($\delta_D$), multiplicative scale factor for covariance matrices ($\beta_C$), and additive edge probability jump for networks ($\delta_N$). Columns correspond to the metric space; rows to the DGP (Piecewise-Smooth vs. Piecewise-Constant). Results based on 1000 simulations with $N=200$. Dotted horizontal line indicates the nominal test level of $\alpha=0.05$.}
\label{fig:all_power_curves_combined_stacked}
\end{figure}

 \section{Empirical Illustrations} \label{sec:applications}

 \subsection{Non-Compete Agreements and the Composition of Work from Home}

\paragraph{Question.} Here, I apply the test to investigate how the enforceability of non-compete agreements (NCAs) affects the composition of work-from-home (WFH) arrangements. In 2020, Washington state enacted a law (RCW 49.62) that makes NCAs enforceable only for employees whose annual earnings exceed a specific, inflation-adjusted threshold of approximately \$100,000. This policy creates a sharp discontinuity where non-compete agreements are unenforceable below the income threshold but become enforceable when an employee's salary crosses it. Unlike a traditional RD design, however, I cannot rule out manipulation around the threshold. Indeed, the object of interest is precisely how firms and workers strategically adjust employment terms when the legal landscape of post-employment mobility changes. 

My central hypothesis is that WFH arrangements become a key bargaining margin, particularly during the sample period which covered the COVID-19 pandemic years when WFH became ubiquitous. When an employee's salary crosses the threshold, rendering an NCA enforceable, their ability to switch to a local competitor is curtailed. One might expect a worker to respond by seeking employment with an out-of-state firm to escape the policy's reach. However, my results below do not reveal any significant jump in the probability of working for an out-of-state employer at the threshold. This suggests that adjustments are more likely to occur within the current employment relationship. In particular, faced with a binding NCA, a high-earning employee may demand non-monetary compensation for their reduced future mobility. From the firm's perspective, granting greater WFH flexibility could be a valuable, low-cost concession to attract these employees. Therefore, I test for an abrupt shift in the average WFH composition at the earnings threshold for newly hired employees.

Previous literature has found that only 10\% of employees bargain over non-compete agreements \citep{starr2021noncompete}, though this was estimated prior to the spread of non-compete laws such as Washington's, and likely increases for higher-income employees like the ones in my sample. \citet{lipsitz2022low} found that non-competes negatively affect workers by reducing their wages and job-to-job mobility, creating cause for bargaining over the agreement. Further, the emerging literature on the effects of remote work have documented WFH's role as a job amenity that is substitutable for higher income and leads to less attrition \citep{mas2017valuing, barrero2022shift, bloom2022hybrid, cullen2025home}.

\paragraph{Data and Empirical Approach.} I use data from the US Census Bureau's Survey of Income and Program Participation (SIPP) for the years 2020 through 2023 to construct a panel of individual job spells for employees in Washington state. The running variable is the employee's annual salary, centered around the relevant year-specific legal threshold. The outcome of interest is the composition of an employee's WFH schedule, defined as a 6-dimensional vector representing the average proportion of time during the job spell that the employee worked 0 to 5 days from home. Since the components of this compositional vector are non-negative and sum to one, I formally test for a discontinuity by applying a square-root transformation so the vector lies on the positive orthant of a unit 5-sphere ($\mathbb{S}^5 \cap \mathbb{R}_+^6$) and adopting the geodesic metric, $d_g\left(y_1, y_2\right)=\arccos \left(y_1^{\top} y_2\right)$ \citep[Supp. S.4.2]{bhattacharjee2023single}.

\begin{figure}[ht!]
\includegraphics[width=0.8\textwidth]{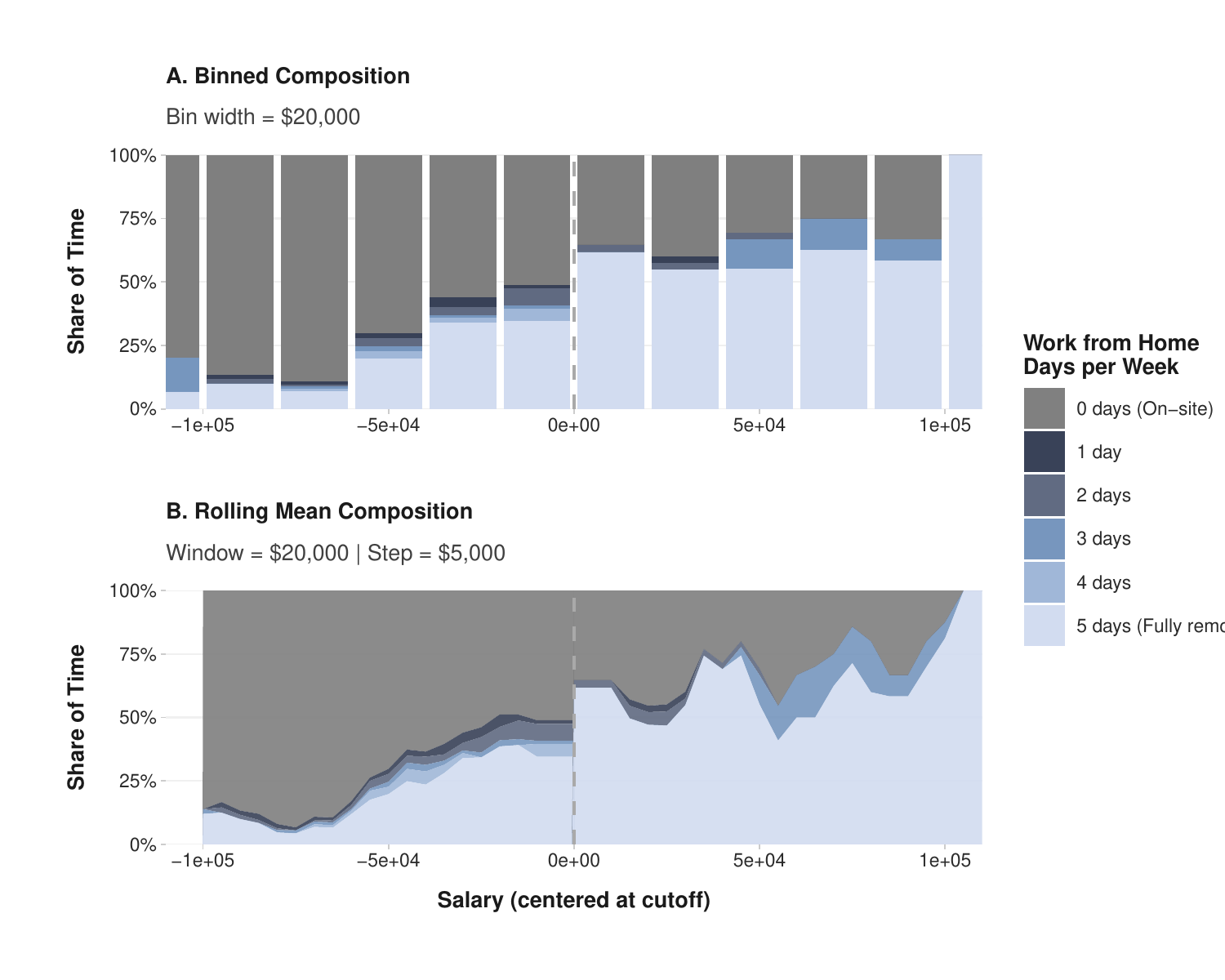}
\caption{Work from Home Composition by Income: Binned and Smoothed}
\label{fig:wfh_descriptive}
\floatfoot{\scriptsize \textit{Note}: figures show 100\% stacked bar chart and 100\% stacked area chart indicating percentage of time spent working from home 0--5 days a week throughout a given employment spell (y axis), by annual salary of the employee (x axis). Annual salary is centered around the Washington non-compete enforceability cutoff in each respective year (2020-2023), indicated by dashed vertical line. Darker colors indicate fewer days worked from home. }
\end{figure}

\paragraph{Results.}

Figure \ref{fig:wfh_descriptive} depicts the evolution of employees' WFH composition during their job spell (y axis) by their annual salary (x axis), using both a binned stacked bar chart and one smoothed in a rolling window. The annual salary variable was centered around the enforceability cutoff of approximately \$100K annual income. In both panels of the figure, one can see a stark jump in the share of days worked fully remotely (light blue color) once the non-compete agreements become enforceable, while the share of time worked fully in person or partly at home diminishes. This suggests that employees, when faced with the possibility of a non-compete, used work from home as a bargaining margin during the pandemic. 

\begin{figure}[ht!]
\includegraphics[width=0.8\textwidth]{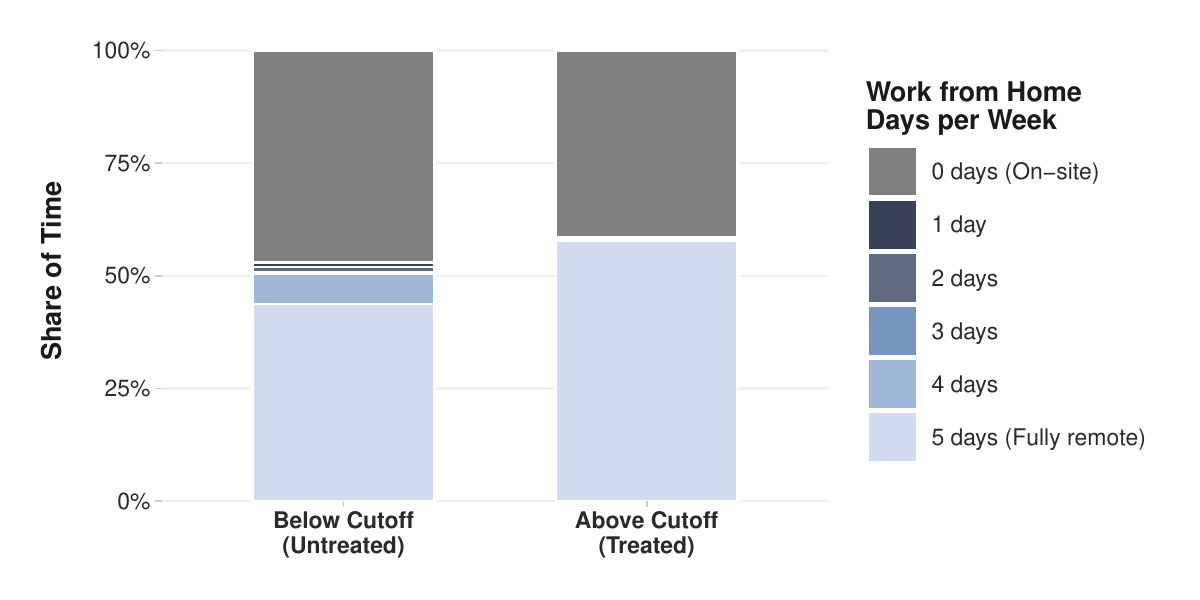}
\caption{Non-Compete Enforceability and Work-From Home Job Composition: Fr\'echet Mean Estimates}
\label{fig:frechet_mean}
\floatfoot{\textit{Note}: figure shows estimates of Fr\'echet mean of employees' WFH composition to the left and right of the non-compete enforceability cutoff in Washington state (2020-2023). Darker colors indicate fewer days worked from home.}
\end{figure}

These descriptive results are echoed in the  local Fr\'echet mean estimates to the left and right of the cutoff in Figure \ref{fig:frechet_mean}. The mean WFH composition below the cutoff has approximately 15\% less time spent working fully remotely than the composition above the cutoff. Partly remote work vanishes completely, while fully in-person work diminishes by approximately 5\%. The test for this difference between the two composition vectors has a p-value of 0.059, significant at the 10\% level. By contrast, Table \ref{tab:rdd} reports the coefficients of a classical local polynomial regression on each of the components of the vector separately (share of time worked 0--5 days at home, as well as whether any or all time was spent working from home). None of these coefficients are significant at the 10\% level, illustrating the additional power the test can provide by exploiting the  information in the compositional structure of the data. As mentioned, I do not find evidence that employees above the cutoff were more likely to work or move out of state. Furthermore, although one may expect employees to respond strategically to the cut-off point, the RD estimates do not suggest any covariate imbalance in age, gender, or education of the employees. These findings indicate that the enforceability of non-competes induced a structural shift in WFH arrangements, which served as a key bargaining margin for affected employees during the pandemic.

\subsection{Preferential Tariff Loss and Countries' Input-Output Networks}

\paragraph{Question.} 
Next, I investigate how national economies restructure in response to losing preferential trade status with the US. To that end, I exploit a unique policy event following the reauthorization of the US Generalized System of Preferences (GSP) on June 29, 2015. The GSP program, which grants duty-free access for thousands of products, had been inactive for nearly two years prior to this date. The 2015 reboot created a clear, forward-looking timeline for firms and governments: a country's GSP eligibility for the year 2017 would be determined by its Gross National Income (GNI) per capita from the year 2015.

This setup creates a regression discontinuity design for studying the effect of preferential tariffs on a country's production network. The running variable is a country's 2015 GNI per capita. The treatment is the loss of GSP benefits on January 1, 2017, for countries whose 2015 GNI crossed the high-income threshold. The question is whether this negative trade shock induced a structural change in a country's national input-output (I-O) network as export demand dropped and firms adapted their supply chains.

A potential concern is that the World Bank's high-income threshold might trigger other policies concurrently with the loss of GSP benefits. However, the most significant of these—the loss of Official Development Assistance (ODA) and graduation from the World Bank's IBRD lending—operate on longer timelines. ODA eligibility ends only after a country remains high-income for three consecutive years \citep{oecd_oda_2023}, while IBRD graduation is a gradual, case-by-case process that unfolds over several years \citep{wb_ibrd_grad_2011}. In contrast, US GSP status is revoked without exception on a sharp two-year trigger \citep{crs_gsp_2017}. This ensures that at the point of GSP expiration, other major policy changes linked to the income threshold have not yet occurred, isolating the GSP effect. 

The international trade literature has documented extensive direct and indirect links between foreign exports, tariffs, and domestic production networks \citep{baqaee2024networks, Dhyne2021Trade, AmitiKonings2007Trade}. Specifically for the GSP, \citet{ozden2005perversity} found that GSP graduation correlates with affected countries implementing more liberal trade policies, while \citet{hakobyan2020gsp} found that GSP expiration led to sizeable declines in exports to the US of around 3\%. However, to my knowledge, there exist no empirical methods to study the impact of trade shocks directly on the full structure of production networks. This underscores the usefulness of the proposed test for empirically validating predictions from trade network models and evaluating the impact of trade shocks.

\paragraph{Data and Empirical Approach.}
The running variable, GNI per capita, is sourced from the World Bank Development Indicators (WDI), with the cutoff set at the official 2015 high-income threshold. The primary outcomes are domestic input-output (I-O) networks for 2016 (capturing anticipatory effects) and 2017 (capturing implementation effects), derived from the EORA26 global database \citep{lenzen2012mapping}. The analysis is constrained to these years because the GSP program lapsed again at the end of
2017, and the subsequent period is further confounded by the COVID-19 pandemic. For each of the 140 countries in the final dataset, I extract the domestic I-O table from the EORA26 database \citep{lenzen2012mapping} and aggregate its 26 sectors into 8 broader categories (see Table \ref{tab:eora_agg_map_short}). The outcome objects for the test are the Leontief Inverse (or total requirements) matrices derived from these $8\times 8$ I-O matrices. A Leontief Inverse captures the total (direct and indirect) output from each sector required to satisfy one unit of final demand, thus representing the economy's full production multipliers. 

To perform the test, I work on the space of graph Laplacians of these production networks, equipped with the Frobenius norm, which satisfies the necessary metric-space assumptions \citep{zhou2022network}. For the main analysis, I symmetrize the Laplacians, though results are highly similar when using in-degree or out-degree Laplacians. Critically, because the population Fr\'echet mean of graph Laplacians under the Frobenius norm is the entrywise average of their matrix elements \citep{zhou2022network}, the test directly captures changes in the country-average production multiplier network.

\paragraph{Results.}

The test indicates a statistically significant jump in the conditional Fr\'echet mean of the Leontief Inverse matrices at the GSP graduation threshold ($p=0.028$). This suggests that countries losing GSP access experience a structural change in their domestic production networks. Moreover, the test also shows a significant change in the average Leontief Inverse matrices one year before GSP graduation ($p=0.016$), further suggesting that countries already began to adjust their production networks in anticipation of the policy change. Crucially, a placebo test for 2015—a year when the GSP program was inactive—reveals no statistically significant effect ($p=0.290$), reinforcing the conclusion that the detected changes are indeed attributable to the GSP graduation event and not to other confounding policies.

Looking at the jump estimates, Figure~\ref{fig:frechet_diff_lagged_rv_2015_io_lag_2} plots the elementwise difference between the treated and control Fr\'echet means of the $8\times8$ Leontief Inverse matrices in 2017 (mapped back from the graph Laplacians). Warmer colors indicate stronger cross-sector multipliers in the treated group while cooler colors indicate weaker ones. The most striking changes appear in the resource-export chain. The effect magnitudes represent the change in total input required per dollar of final demand. The Mining-Manufacturing flows fall by 0.47 and the Mining-Trade ones by 0.36 on average, while Manufacturing's spill-overs into Construction and Trade decline by 0.26 and 0.21 respectively. These contractions indicate that, once preferential tariffs disappear, raw materials are no longer routed through domestic factories and wholesale networks for export processing.

\begin{figure}[ht!]
\includegraphics[width=0.7\textwidth]{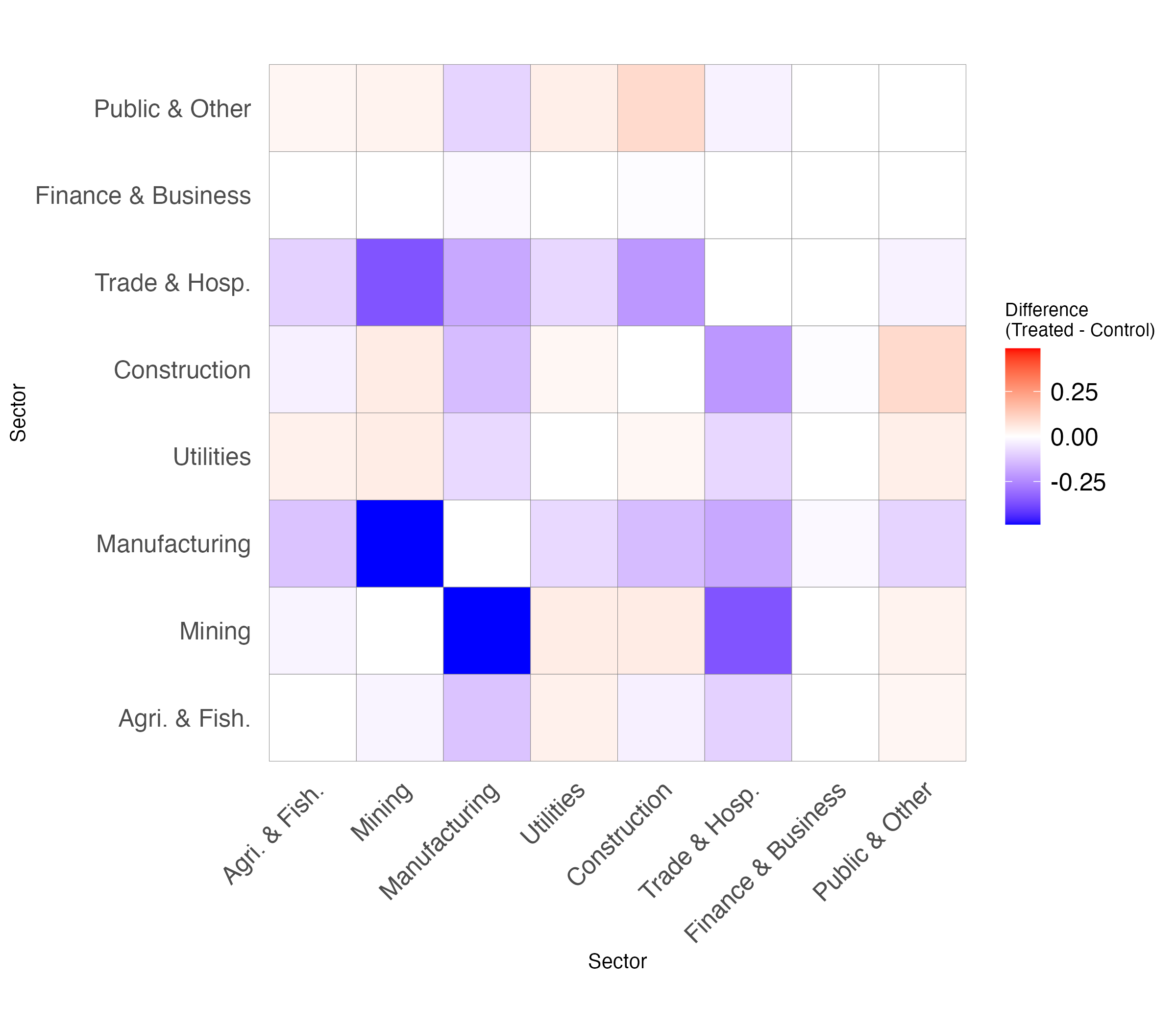}
\caption{Jump in Input-Output Networks After Preferential Tariff Loss}
\label{fig:frechet_diff_lagged_rv_2015_io_lag_2}
\floatfoot{\textit{Note:} figure shows estimates of difference in Fr\'echet means of countries' Leontief Inverse matrices to the left and right of the US GSP eligibility cutoff.}
\end{figure}

A different pattern emerges on the domestic side of the economy. Construction's linkage with Public and Other Services rises by 0.11, and Mining's ties with both Construction (0.06) and Utilities (0.04) also increases slightly.  The Utilities industry itself deepens its connections with Construction and the public sector (about 0.03 each). Together, these increases suggest that extraction industries redirect output towards home-market infrastructure and energy projects, while construction firms rely more heavily on government demand, likely to offset the export shortfall.

Table~\ref{tab:linkage_ranks} lists the five largest absolute decreases and increases.  All major contractions involve Mining or Manufacturing on at least one end, whereas four of the five strengthened links involve Construction or Utilities. The overall heatmap in Figure \ref{fig:frechet_diff_lagged_rv_2015_io_lag_2} is light blue along the main diagonal, implying a mild reduction in own-sector multipliers as well. Taken together, these patterns point to a sparser and less tightly coupled production network in economies that lose GSP access—an adjustment that would be invisible in most scalar outcomes such as total trade or sectoral output alone. Indeed, in Table \ref{tab:io_scalar_rd}, I report the results of a standard RD analysis on scalar network outcomes such as centrality, complexity, and multiplier effects. None of these show significant changes at the GSP graduation threshold, highlighting the value of the proposed Fr\'echet jump test for capturing network-wide restructuring that would otherwise be missed.

\section{Conclusion} \label{sec:conclusion}
This article has introduced a novel test for detecting discontinuities in conditional Fr\'echet means within general metric spaces. By leveraging local Fr\'echet regression techniques, it develops test statistics that compare the conditional Fr\'echet variances to the left and right of a potential jump location. For that, it establishes a central limit theorem for the conditional Fr\'echet variance, derived the asymptotic null distribution of the proposed test statistic, and demonstrated its consistency against relevant alternatives. Two empirical applications illustrate the utility of the test by providing novel evidence for the effect of non-competes on work from home job compositions and of preferential tariffs on developing countries' production networks.

The proposed framework is general, making it applicable to a wide range of data types where outcomes are elements of a metric space, such as distributions, networks, or spheres, that evolve nonparametrically in a scalar covariate. The test provides a formal tool for investigating abrupt changes in the central tendency of such complex data. This can be valuable in many fields, including economics, biology, and image analysis, where identifying structural breaks or treatment effects in non-Euclidean data is of increasing importance. Future research could explore finite sample performance, optimal bandwidth selection, and extensions to multivariate conditioning variables or other forms of Fr\'echet regression.

\bibliographystyle{apalike}
\bibliography{references} 

\clearpage
\appendix
\begin{supplement}

% =============================================================================
% TABLE OF CONTENTS FOR THE SUPPLEMENT
% =============================================================================
\par\noindent
\hrule
\par\medskip\noindent
This supplement contains additional details, proofs, and results. The sections are:
\begin{itemize}
    \item Appendix~\ref{app:frechet_details}: \hyperref[app:frechet_details]{Details on One-Sided Local Fr\'echet Regression Estimators}
    \item Appendix~\ref{app:bandwidth}: \hyperref[app:bandwidth]{Bandwidth Selection Procedure}
    \item Appendix~\ref{app:simulation}: \hyperref[app:simulation]{Simulation Details}
    \item Appendix~\ref{app:proofs}: \hyperref[app:proofs]{Proofs of Theoretical Results}
    \item Appendix~\ref{app:add_results}: \hyperref[app:add_results]{Additional Results}
\end{itemize}
\hrule
\par\bigskip
% =============================================================================

\section{Local Fr\'echet Regression}
\label{app:frechet_details}

This appendix provides a more detailed overview of the one-sided local linear Fr\'echet regression estimators discussed in Section~\ref{sec:twosidedmonster}, adapting the estimators in \citep{petersen2019frechet} to estimation at a jump point.

The goal is to estimate the conditional Fr\'echet mean at the boundary $c$ from the left and right sides. The one-sided local Fr\'echet regression estimators are defined as the minimizers of a locally weighted sum of squared distances,
\begin{equation}
\hat{l}_{\pm,\oplus} \coloneqq \underset{\omega \in \Omega}\argmin \: \left\{ \hat{L}_{\pm,n}(\omega) \coloneqq  n^{-1} \sum_{i=1}^n s_{\pm,in}(c,h) d^2\left(Y_i, \omega\right) \right\},
\end{equation}
where $h$ is a bandwidth and $s_{\pm,in}(c,h)$ are one-sided local linear regression weights.

Specifically, the weights for the right-sided estimator ($\hat{l}_{+, \oplus}$) are given by,
\begin{align*}
s_{+,in}(c,h) &= \mathrm{1}(X_i \geq c) \frac{K_h(X_i-c)}{\hat{\sigma}_{+,0}^2} \left[ \hat{\mu}_{+,2} - \hat{\mu}_{+,1}(X_i-c) \right], \\
\intertext{where $K_h(u) = K(u/h)/h$ for a kernel function $K$, and the components are sample moments calculated using only data to the right of the cutoff,}
\hat{\mu}_{+,j} &= n^{-1} \sum_{k=1}^n \mathrm{1}(X_k \geq c) K_h(X_k-c)(X_k-c)^j \quad \text{for } j=0,1,2, \\
\hat{\sigma}_{+,0}^2 &= \hat{\mu}_{+,0}\hat{\mu}_{+,2} - \hat{\mu}_{+,1}^2.
\end{align*}
The weights for the left-sided estimator, $s_{-, i n}(c, h)$, are defined analogously using observations where $X_i < c$ and $X_k < c$.

The corresponding population (or pseudo-true) counterparts of these estimators are defined as,
\begin{equation} \label{eq:l_tilde_pm_app}
\tilde{l}_{\pm, \oplus}\coloneqq \underset{\omega \in \Omega}{\operatorname{argmin}} \: \left\{ \tilde{L}_{\pm,n}(\omega) \coloneqq  E\left[s_\pm(X, c, h) d^2(Y, \omega)\right] \right\},
\end{equation}
where $s_\pm(X,c,h)$ are the population versions of the weights. For the right side, these are,
\begin{align*}
s_+(X,c,h) &= \mathrm{1}(X \geq c) \frac{K_h(X-c)}{\sigma_{+,0}^2} [\mu_{+,2} - \mu_{+,1}(X-c)], \\
\intertext{with population moments,}
\mu_{+,j} &= E[\mathrm{1}(X \geq c)K_h(X-c)(X-c)^j] \quad \text{for } j=0,1,2, \\
\sigma_{+,0}^2 &= \mu_{+,0}\mu_{+,2} - \mu_{+,1}^2.
\end{align*}
Again, the left-sided population weights $s_-(X,c,h)$ and moments $\mu_{-,j}$ are defined analogously over the domain $X < c$.

\section{Bandwidth Selection Procedure}
\label{app:bandwidth}

To select the bandwidth for the Fr\'echet mean estimation, I employ a K-fold cross-validation (CV) procedure that aims to minimize the out-of-sample prediction error at the cutoff. This approach adapts the general principles of cross-validation for local polynomial regression \citep{fan1996local} to the one-sided setting, similar to a canonical approach in regression discontinuity design \citep{imbens2008regression}. For the specific case of Wasserstein-Fr\'echet regression, one can derive an explicit expression for the optimal bandwidth in terms of the regression's mean-squared error by leveraging the connection to quantile estimation \citet{van2025regression}. In the more general metric space setting, this is not possible due to the lack of a Taylor series approximation. As an alternative to the bandwidth selection procedure proposed here, one could also use the heuristic procedure in \citet{kurisu2025regression}, which targets smoothness of the estimated one-sided regression functions near the cutoff.

The goal of the cross-validation is to find the bandwidth $h$ that minimizes the sum of squared prediction errors. Let the full dataset be $\mathcal{D} = \{(Y_i, X_i)\}_{i=1}^n$. Partition the indices $\{1, \dots, n\}$ into $K$ disjoint folds, denoted $\mathcal{I}_1, \dots, \mathcal{I}_K$. For each fold $k$, define a training set $\mathcal{D}^{(-k)} = \mathcal{D} \setminus \mathcal{D}^{(k)}$ and a validation set $\mathcal{D}^{(k)}$.

For a given candidate bandwidth $h$, estimate the one-sided Fr\'echet means at the cutoff, $\hat{l}_{-, \oplus}^{(-k)}(c; h)$ and $\hat{l}_{+, \oplus}^{(-k)}(c; h)$, using only the training data $\mathcal{D}^{(-k)}$. The cross-validation error for this bandwidth is then calculated as the sum of squared distances between the observations in the validation fold $\mathcal{D}^{(k)}$ and the corresponding Fr\'echet mean estimated from the training data:
\begin{equation}
\text{CV}(h) = \sum_{k=1}^K \left( \sum_{i \in \mathcal{I}_k, X_i < c} d^2\big(Y_i, \hat{l}_{-, \oplus}^{(-k)}(c;h)\big) + \sum_{i \in \mathcal{I}_k, X_i \ge c} d^2\big(Y_i, \hat{l}_{+, \oplus}^{(-k)}(c;h)\big) \right).
\end{equation}
The optimal bandwidth is the one that minimizes this objective function: $\hat{h}_{\text{CV}} = \argmin_{h \in \mathcal{H}} \text{CV}(h)$, where $\mathcal{H}$ is a grid of candidate bandwidths.

 The candidate bandwidth grid $\mathcal{H}$ is constructed based on data-driven heuristics, using the range and minimum spacing of the covariate values on each side of the cutoff. To preserve the structure of the one-sided estimation problem, the K-fold partitioning is stratified by the cutoff: the observations with $X_i < c$ and $X_i \geq c$ are each partitioned into $K$ folds.

 For the empirical applications, I use $K=5$ folds and a grid of $n_{\text{bw}}=10$ candidate bandwidths.

\section{Simulation Details} \label{app:simulation}

For the three metric spaces considered, I specify the following data-generating processes:
\begin{itemize}
    \item \textbf{Density Space (2-Wasserstein metric):} The random objects $Y_i$ are generated as quantile functions corresponding to Gaussian distributions, $N(\mu(x), \sigma_0^2)$, with $\sigma_0^2=1$. The Fr\'echet mean of such distributions under the 2-Wasserstein distance is the Gaussian with the average mean. Therefore, a jump in $\mu(x)$ directly induces a jump in the Fr\'echet mean. We represent quantile functions on an equidistant grid of 50 points in $(0,1)$.
    \begin{itemize}
        \item \textit{Piecewise-Smooth:} The mean follows $\mu(x) = 0.8(x-0.5)$. Under $H_1$, the mean for observations with $X_i \ge c$ is shifted to $\mu(x) = 0.8(x-0.5) + \delta_D$.
        \item \textit{Piecewise-Constant:} The baseline mean is $\mu(x) = 0.5$. Under $H_1$, the mean for $X_i \ge c$ becomes $\mu(x) = 0.5 + \delta_D$.
        \item The jump magnitude $\delta_D$ is varied symmetrically from $-2.0$ to $2.0$.
    \end{itemize}

    \item \textbf{Covariance Matrix Space (Frobenius metric):} Objects are $3 \times 3$ sample covariance matrices, each estimated from 300 draws from a multivariate normal distribution $N(0, \Sigma(x))$. The Fr\'echet mean of these sample covariance matrices under the Frobenius norm converges to the population covariance matrix $\Sigma(x)$.
    \begin{itemize}
        \item \textit{Piecewise-Smooth:} The baseline covariance matrix $\Sigma_0(x)$ has diagonal elements $1.5 + 0.6(x-0.5)$ and primary off-diagonal element $0.2 + 0.3(x-0.5)$.
        \item \textit{Piecewise-Constant:} The baseline covariance matrix $\Sigma_0(x)$ has diagonal elements of $1.5$ and a primary off-diagonal of $0.2$.
        \item Under $H_1$, for $X_i \ge c$, the covariance matrix is scaled by a factor $\beta_C$. Specifically, $\Sigma(x) = D_\beta \Sigma_0(x) D_\beta$, where $D_\beta$ is a diagonal matrix with $\sqrt{\beta_C}$ on its diagonal. This corresponds to scaling all standard deviations by $\sqrt{\beta_C}$. The factor $\beta_C$ varies from $0.5$ to $2.0$.
    \end{itemize}

    \item \textbf{Network Space (Frobenius metric on Laplacians):} Objects are graph Laplacians of $10 \times 10$ undirected graphs. For each pair of nodes $(u,v)$, an edge is generated independently from a Bernoulli distribution with probability $p(x)$. The Fr\'echet mean of the resulting Laplacians under the Frobenius norm is the Laplacian of the graph whose adjacency matrix has entries equal to the mean edge probabilities $p(x)$. Thus, a jump in $p(x)$ creates a jump in the Fr\'echet mean.
    \begin{itemize}
        \item \textit{Piecewise-Smooth:} The edge probability is $p(x) = 0.4 + 0.2(x-0.5)$. Under $H_1$, for $X_i \ge c$, the probability becomes $p(x) = 0.4 + 0.2(x-0.5) + \delta_N$.
        \item \textit{Piecewise-Constant:} The baseline probability is $p(x) = 0.4$. Under $H_1$, for $X_i \ge c$, it becomes $p(x) = 0.4 + \delta_N$.
        \item The edge probability $p(x)$ is truncated to lie in $[0.05, 0.95]$. The jump magnitude $\delta_N$ varies from $-0.3$ to $0.3$.
    \end{itemize}
\end{itemize}

\section{Proofs} \label{app:proofs}

In all proofs that follow, I focus on the argument for the estimator ``from the right'' to avoid repetition or excessive notation, $m_{+,\oplus}$. The arguments ``from the left'' follow analogously. Further, I suppress the $c$ argument of $s_{\pm,in}(c,h)$ to ease notation, since it is fixed at $c$ everywhere.

\subsection{One-Sided Local Fr\'echet Convergence Results}
\label{app:frechet}
Here, I work out the proofs for the convergence of the ``one-sided'' version of local Fr\'echet regression in detail, adapting the proofs in \citet{petersen2019frechet}. The reason for working them out completely is that I will reuse parts in the proofs for the jump test in Appendix \ref{app:anova_proofs} below. 
For notational ease, I drop the $c$ argument in the objective functions $\hat{L}_{+,n}(\omega, c), \tilde{L}_{+,n}(\omega, c)$.

For completeness, I include the following variation of a well-known result \citep{fan1996local}. Define \[\tau_{+,j}(y) = E\left[ \mathrm{1} (X\geq x) K_h(X-x)(X-x)^j \mid Y = y \right]\] and remember the definition of $\mu_{+,j} = E\left[ \mathrm{1} (X\geq x) K_h(X-x)(X-x)^j \right]$, the corresponding estimator $\hat{\mu}_{+,j} = \frac1n \sum_{i=1}^n \mathrm{1}(X_i \geq x) K_h(X_i - x)(X_i - x)^j$ for $j=0,1,2$, and $K_{\pm, k j}=\int_{\mathcal{R}} 1(u \underset{<}{\geq} 0) K^k(u) u^j \mathrm{~d} u$.
 \begin{lemma} \label{lemma:fan_gijbels}
Suppose Assumptions \ref{asspt:kernel}, \ref{asspt:sampling}, and \ref{asspt:densities} hold. Then at $x=c$, we have,
\[
\mu_{+,j}=h^j\left[f_X(c) K_{+,1 j}+h f_X^{\prime}(c) K_{+,1(j+1)}+O\left(h^2\right)\right]
\]
and $\hat{\mu}_{+,j} = \mu_{+,j} + O_p\left(\left(h^{2 j-1} n^{-1}\right)^{1 / 2}\right) \text { for } j\in\mathbb{N}_0$. Additionally, for $x \ge c$,
\[
\tau_{+,j}(y) = h^j\left[g_{y}(c) K_{+,1 j}+h g_{y}^{\prime}(c) K_{+,1(j+1)}+O\left(h^2\right)\right],
\]
with the $O(h^2)$ term holding uniformly over $y \in \Omega$.
\end{lemma}

\begin{proof}
The equalities for $\mu_{+,j}$ and $\tau_{+,j}(y)$ follow from Assumption \ref{asspt:densities}, which guarantees that $f_X(x)$ and $g_y(x)$ are twice continuously differentiable for $x$ in a right-sided neighborhood of $c$. We can therefore apply a second-order Taylor expansion for $f_X$ and $g_y$ around $c$. Combining this with two changes of variables, $u=z-c$ and then $v=u/h$, yields the stated expressions for the population moments.

For the sample moments, $E[\hat{\mu}_{+,j}] = \mu_{+,j}$ holds by definition. For the variance,
\[
\operatorname{Var}(\hat{\mu}_{+,j}) = \frac{1}{n} \operatorname{Var}\left(\mathrm{1}(X_i \geq c) K_h(X_i-c)(X_i-c)^j\right) \le \frac{1}{n} E\left[\mathrm{1}(X_i \geq c) K_h^2(X_i-c)(X_i-c)^{2j}\right].
\]
By a change of variables $u=(X_i-c)/h$, the expectation term is:
\[
h^{2j-1} \int_0^\infty K(u)^2 u^{2j} f_X(c+hu) du = O(h^{2j-1}).
\]
Therefore, $\operatorname{Var}(\hat{\mu}_{+,j}) = O((nh^{1-2j})^{-1})$, which implies $\hat{\mu}_{+,j} - \mu_{+,j} = O_p((nh^{1-2j})^{-1/2})$.
\end{proof}
\begin{lemma}\label{lemma:convergence_tilde}
Suppose \ref{asspt:sampling}--\ref{assptT:J}
hold.  Then
\[
  d\!\bigl(\tilde l_{+,\oplus},m_{+,\oplus}\bigr)
     =O\!\bigl(h^{2/(\beta_1-1)}\bigr),
\]
and the same rate holds for $\tilde l_{-,\oplus}$.
\end{lemma}

\begin{proof}
For notational convenience write $Y$ for the outcome random element.
We first show that
\begin{equation} \label{eq:donttagme}
  \frac{dF_{Y\mid X=x}(y)}{dF_Y(y)}=\frac{g_y(x)}{f_X(x)},
  \qquad x\in(c-\varepsilon,c+\varepsilon).
\end{equation}
Let \(U\subset\Omega\) be open and define
\[
  a(x)=\int_U \frac{g_y(x)}{f_X(x)}\,dF_Y(y),
  \qquad
  b(x)=\int_U dF_{Y\mid X}(y\mid x).
\]
Assumption \ref{asspt:densities}(c) implies that \(a(\cdot)\) and
\(b(\cdot)\) are continuous on \((c-\varepsilon,c+\varepsilon)\).
Hence, for any \(z\) in that interval,
\[
\begin{aligned}
\int_{c-\varepsilon}^{z} a(x)f_X(x)\,dx
 &=\int_U\!\int_{c-\varepsilon}^{z} g_y(x)\,dx\,dF_Y(y)              \\
 &=\int_U\!\int_{c-\varepsilon}^{z} dF_{X\mid Y}(x\mid y)\,dF_Y(y)   \\
 &=\int_{(c-\varepsilon,z)\times U} dF_{X,Y}(x,y)                    \\
 &=\int_{c-\varepsilon}^{z} b(x)f_X(x)\,dx ,
\end{aligned}
\]
which yields \eqref{eq:donttagme}.

Using Lemma \ref{lemma:fan_gijbels},
\(\mu_{+,j}=h^j\{f_X(c)K_{+,1j}+O(h)\}\) and
\(\sigma_{+,0}^2
   =h^{2}f_X(c)^2\!\bigl(K_{+,10}K_{+,12}-K_{+,11}^2\bigr)\{1+O(h)\}\). Thus, for every \(y\in\Omega\),
$$
\begin{aligned}
& \int s_{+}(z, c, h) \mathrm{d} F_{X \mid Y}(z, y) \\
& \quad=\int 1(z \geq c) \frac{K_h(z-c)}{\sigma_{+, 0}^2}\left(\mu_{+, 2}-\mu_{+, 1}(z-c)\right) \mathrm{d} F_{X \mid Y}(z, y) \\
& \quad=\frac{\mu_{+, 2} \tau_{+, 0}(y)-\mu_{+, 1} \tau_{+, 1}(y)}{\sigma_{+, 0}^2} \\
& \quad=\frac{h^2 g_{y}(c)\left(f_X(c)\left(K_{+, 12} K_{+, 10}-K_{+, 11}^2\right)+h f^{\prime}(c)\left(K_{+, 10} K_{+, 13}-K_{+, 11} K_{+, 12}\right)\right)+O\left(h^4\right)}{h^2 f_X(c)\left(f_X(c)\left(K_{+, 12} K_{+, 10}-K_{+, 11}^2\right)+h f^{\prime}(c)\left(K_{+, 10} K_{+, 13}-K_{+, 11} K_{+, 12}\right)\right)+O\left(h^4\right)} \\
& \quad=\frac{g_{y}(c)}{f_X(c)}+O\left(h^2\right)
\end{aligned}
$$
Consequently,
\[
\begin{aligned}
\tilde L_{+,n}(\omega,c)
  &=\iint s_{+}(z,c,h)\,d^{2}(y,\omega)\,dF_{X,Y}(z,y) \\[2pt]
  &=\int d^{2}(y,\omega)\!
        \Bigl[\int s_{+}(z,c,h)\,dF_{X\mid Y}(z\mid y)\Bigr]
        dF_Y(y) \\[2pt]
  & = 
     \int d^{2}(y,\omega)
          \Bigl\{\tfrac{g_y(c)}{f_X(c)}+O(h^{2})\Bigr\}
          dF_Y(y) \\[2pt]
  &=\lim_{x \to c^+}\!
     \int d^{2}(y,\omega)\tfrac{g_y(x)}{f_X(x)}\,dF_Y(y)   %  ⟨—  justification
     \;+\;O(h^{2})                                         \\[-2pt]
  &=\lim_{x \to c^+}E[d^{2}(Y,\omega)\mid X=x]\;+\;O(h^{2}) \\[2pt]
  &=M_{+,\oplus}(\omega)+O(h^{2}).
\end{aligned}
\]
where the third equality follows from the approximation derived above and the fourth equality follows from continuity of $x\mapsto g_y(x)/f_X(x)$ and bounded $d^{2}(y,\omega)$ which allows the use of the dominated convergence theorem. By \ref{asspt:uniqueness}, this implies that $d\left(m_{+, \oplus}, \tilde{l}_{+, \oplus}\right)=o(1) \text { as } h=h_n \rightarrow 0$. The rest of the proof follows by a similar argument as in the proof of Theorem 3 in \citet{petersen2019frechet}. I include it below for completeness.

From the penultimate line in the derivation above, we have,
$V(\omega) \coloneqq \tilde{L}_{+,n}(\omega, c) - M_{+,\oplus}(\omega, c) = \int O(h^2) d^2(y,\omega) \dd F_{Y}(y)$. As a result,
\begin{align}
\sup_{d(\omega, m_{+,\oplus}) < \delta } \left| V(\omega) - V(m_{+,\oplus}) \right| & =\sup_{d(\omega, m_{+,\oplus}) < \delta } \left| \int O(h^2) \left( d^2(y, \omega) - d^2(y, m_{+,\oplus}) \right) \dd F_{Y}(y) \right| \nonumber \\ 
& \leq a h^2 2 \operatorname{diam}(\Omega) \delta \coloneqq b \delta h^2 \label{eq_app:E_sup}
\end{align}
for some $a, b > 0$, where the last line follows from the triangle inequality. Then, the remainder of the proof follows as in \citet[Theorem 3.2.5]{vaart1996weak}. Let $r_h=h^{-\frac{\beta_1}{\beta_1-1}}$ so that the rate inequality in that theorem holds. Also, define $S_{j, n}=\left\{\omega: 2^{j-1}<r_h d\left(\omega, m_{+,\oplus}(x)\right)^{\beta_1 / 2} \leq 2^j \right\} $. Then, choose $\eta$ to satisfy \ref{assptT:J} and small enough such that \ref{asspt:curvature} \ref{asspt:curvature_population} holds for all $\delta < \eta$. Set $\tilde{\eta} = \eta^{\frac{\beta_1}{2}}$. Then for any integer $M$, 
\begin{align*}
& \mathrm{1} \left(r_h  d\left(\tilde{l}_{+,\oplus}, m_{+,\oplus}\right)^{\beta_1 / 2}>2^M\right) \leq \mathrm{1}\left(2 d\left(\tilde{l}_{+,\oplus}, m_{+,\oplus}\right) \geq \eta\right)  + \\
&  \sum_{\substack{j \geq M \\ 2^j \leq r_h \tilde{\eta}}} \mathrm{1}\left( \sup _{\omega \in S_{j, n}}\left|V(\omega)-V\left(m_{+,\oplus}\right)\right| \geq C \frac{2^{2(j-1)}}{r_h^2} \right),
\end{align*}
where the first term on the right-hand side goes to zero as $n \to \infty$ so that $h_n \to 0$. Further, for each $j$ in the sum in the second term, we have that $d(\omega, m_{+,\oplus}) \leq \left( \frac{2^j}{r_h} \right)^{\frac{2}{\beta_1}} \leq \eta$. As a result, using the inequality $\mathrm{1}\{x \geq a\} \leq \frac{x}{a}$ combined with Eq. \ref{eq_app:E_sup}, we can bound the sum by,
\[
4 b C^{-1} \sum_{\substack{j \geq M \\ 2^j \leq r_h \tilde{\eta}}} \frac{2^{2 j\left(1-\beta_1\right) / \beta}}{r_h^{2\left(1-\beta_1\right) / \beta_1} h^{-2}} \leq 4 b C^{-1} \sum_{j \geq M} \left(\frac{1}{4^{\left(\beta_1-1\right) / \beta_1}}\right)^j
\]
which converges as $\beta_1 > 1$ when $M=M_n \to \infty$. This proves that $d\left(\tilde{l}_{+,\oplus}, m_{+,\oplus}\right)=O\left(r_h^{-2 / \beta_1}\right)=O\left(h^{\frac{2}{\beta_1-1}}\right)$, and by an analogous argument, $d\left(\tilde{l}_{-,\oplus}, m_{-,\oplus}\right)$.  
\end{proof}

\begin{lemma} 
\label{lemma:l_hat-l_tilde}
Suppose \ref{asspt:kernel} and \ref{asspt:uniqueness} hold, $\Omega$ is bounded, $h \to 0$, and $nh \to \infty$. Then $d\left(\hat{l}_{\pm, \oplus}, \tilde{l}_{\pm, \oplus}\right) = o_p(1)$.
\begin{proof}
The result follows an identical argument as in the proof of Lemma 2 in \cite{petersen2019frechet}.
\end{proof}
\end{lemma}

\begin{lemma} \label{lemma:aux_l_hat-l_tilde}
	Under assumptions \ref{asspt:kernel} and \ref{asspt:uniqueness}--\ref{asspt:curvature}, $h\to 0$ and $nh \to \infty$ 
 $$
  d\left(\hat{l}_{\pm, \oplus}, \tilde{l}_{\pm, \oplus}\right) = O_p\left[(n h)^{-1 /\left(2\left(\beta_{2}-1\right)\right)}\right].
$$
\begin{proof}

Again, the proof follows from a similar argument as the proof of Theorem 4 in \citet{petersen2019frechet}, with some small modifications. For the sake of completeness, I include it in detail. 
First, note that since $X_i$ are i.i.d. samples from $X$, we can replace $s_+(X, c, h)$ inside $\tilde{L}_{+,n}$ with 
$$s_{+,i}(h)= \mathrm{1}(X_i \geq c) K_h\left(X_i-c\right) \frac{\mu_{+,2}-\mu_{+,1}\left(X_i-c\right)}{\sigma_{+,0}^2}.$$ 
Define $T_{+,n}(\omega)=\hat{L}_{+,n}(\omega)-\tilde{L}_{+,n}(\omega),$ where I remind the reader that $\hat{L}_{+,n}(\omega) = \frac1n \sum_{i=1} s_{+,in}(h) d^2(Y_i, \omega)$ for a general metric space $(\Omega, d)$. Let
$$
D_i(\omega)=d^2\left(Y_i, \omega\right)-d^2\left(Y_i, \tilde{l}_{+,\oplus}\right)
$$
then we have
\begin{equation}
\label{eq:L_hat-L_tilde}
\begin{aligned}
\left|T_{+,n}(\omega)-T_{+,n}\left(\tilde{l}_{+,\oplus}\right)\right| \leq \left\lvert\, \frac{1}{n} \right. & \left. \sum_{i=1}^n\left[s_{+,i n}(h)-s_{+,i}( h)\right] D_i(\omega) \right| \\
& +\left|\frac{1}{n} \sum_{i=1}^n\left(s_{+,i}(h) D_i(\omega)-E\left[s_{+,i}(h) D_i(\omega)\right]\right)\right|.
\end{aligned}
\end{equation}
Define
$$
W_{+,0 n}=\frac{\hat{\mu}_{+,2}}{\hat{\sigma}_{+,0}^2}-\frac{\mu_{+,2}}{\sigma_{+,0}^2}, \quad W_{+,1 n}=\frac{\hat{\mu}_{+,1}}{\hat{\sigma}_{+,0}^2}-\frac{\mu_{+,1}}{\sigma_{+,0}^2}
$$
and note that $s_{+, in}(h) - s_{+,i}(h) = W_{+,0n} \mathrm{1}(X_i \geq c) K_h(X_i-c) +W_{+,1n} \mathrm{1}(X_i \geq c) K_h(X_i-c) (X_i-c) $. 
From Lemma \ref{lemma:fan_gijbels}, it follows that $W_{+,0 n}=O_p\left((n h)^{-1 / 2}\right)$ and $W_{+,1 n}=O_p\left(\left(n h^3\right)^{-1 / 2}\right)$. Moreover,  \begin{equation} \label{eq:E_k}
\begin{aligned}
E\left[\mathrm{1}(X_i \geq c) K_h\left(X_i-x\right)\left(X_i-x\right)^j\right] & =O\left(h^j\right) \\
E\left[\mathrm{1}(X_i \geq c)  K_h^2\left(X_i-x\right)\left(X_i-x\right)^{2 j} \right] & =O\left(h^{2 j-1}\right)
\end{aligned}
\end{equation}
so $s_{+,in}(h) - s_{+,i}(h) = O_p((nh)^{-1/2})$. 
Combined with the fact that $\left|D_i(\omega)\right| \leq 2 \operatorname{diam}(\Omega) d\left(\omega, \tilde{l}_{+,\oplus}\right)$, the first term on the right-hand side of \eqref{eq:L_hat-L_tilde} is $O_p\left((nh)^{-1/2} d\left(\omega, \tilde{l}_{+,\oplus}\right)\right)$, where the $O_p$ term is independent of $\omega$ and $\tilde{l}_{+,\oplus}$. As a result, we can define
$$
B_{+,R} \coloneqq \left\{\sup _{d\left(\omega, \tilde{l}_{+,\oplus}\right)<\delta}\left|\frac{1}{n} \sum_{i=1}^n\left[s_{+,i n}(h)-s_{+,i}( h)\right] D_i(\omega)\right| \leq R \delta(n h)^{-1 / 2}\right\}
$$
for $R>0$, so that $P\left(B_{+,R}^c\right) \rightarrow 0$.

Next, to control the second term on the right-hand side of \ref{eq:L_hat-L_tilde}, define the functions $g_{+,\omega}: \mathcal{R} \times \Omega \rightarrow \mathcal{R}$ by
$$
g_{+,\omega}(z, y)=\frac{\mathrm{1}(z \geq c)}{\sigma_{+,0}^2} K_h(z-c)\left[\mu_{+,2}-\mu_{+,1}(z-c)\right] d^2(y, \omega)
$$

and the function class,

$$
\mathcal{M}_{+, n \delta}=\left\{g_\omega-g_{\tilde{l}_{+,\oplus}}: d\left(\omega, \tilde{l}_{+,\oplus}\right)<\delta\right\}
$$

An envelope function for $\mathcal{M}_{+, n \delta}$ is

$$
G_{+, n \delta}(z)=\frac{2 \operatorname{diam}(\Omega) \delta}{\sigma_{+,0}^2} \mathrm{1}(z \geq c) K_h(z-c)\left|\mu_{+,2}-\mu_{+,1}(z-c)\right|
$$
and $E\left(G_{+,n \delta}^2(X)\right)=O\left(\delta^2 h^{-1}\right)$. Using this fact together with Theorems 2.7.11 and 2.14.2 of \citet{vaart1996weak} and \ref{assptT:J}, for small $\delta$,

$$
E\left(\sup _{d\left(\omega, \tilde{l}_{+,\oplus}(x)\right)<\delta}\left|\frac{1}{n} \sumn s_{+,i}(h) D_i(\omega)-E\left[s_{+,i}(h) D_i(\omega)\right]\right|\right)=O\left(\delta(n h)^{-1 / 2}\right)
$$

Combining this with \ref{eq:L_hat-L_tilde} and the definition of $B_R$,
\begin{equation} \label{eq:rate_delta_bound}
E\left(\mathrm{1}_{B_{+,R}} \sup _{d\left(\omega, \tilde{l}_{+,\oplus}\right)<\delta}\left|T_{+,n}(\omega)-T_{+,n}\left(\tilde{l}_{+,\oplus}\right)\right|\right) \leq \frac{a \delta}{(n h)^{1 / 2}}
\end{equation}
where $\mathrm{1}_{B_{+,R}}$ is the indicator function for the set $B_{+,R}$ and $a$ is a constant depending on $R$ and the entropy integral in \ref{assptT:J}.

To finish, set $t_n=(n h)^{\frac{\beta_2}{4\left(\beta_2-1\right)}}$ and define
$$
S_{+,jn}=\left\{\omega: 2^{j-1}<t_n d\left(\omega, \tilde{l}_{+,\oplus}\right)^{\beta_2 / 2} \leq 2^j\right\}
$$
The rest of the argument follows similarly to the proof of Lemma \ref{lemma:convergence_tilde} and \citet[Theorem 3.2.5]{vaart1996weak}. Choose $\eta_2$ satisfying \ref{asspt:curvature} \ref{asspt:curvature_estimator} and such that \ref{assptT:J} is satisfied for any $\delta<\eta_2$. Set $\tilde{\eta}:=\left(\eta_2 / 2\right)^{\beta_2 / 2}$. For any integer $M$,
\begin{equation}
\begin{aligned} \label{eq:prob_rate_lhat}
& P\left(t_n d\left(\tilde{l}_{+,\oplus}, \hat{l}_{+,\oplus}\right)^{\beta_2 / 2}>2^M\right) \leq P\left(B_{+,R}^c\right)+P\left(2 d\left(\tilde{l}_{+,\oplus}, \hat{l}_{+,\oplus}\right)>\eta_2\right) \\
& \quad+\sum_{\substack{j \geq M \\
2^j \leq t_n \tilde{\eta}}} P\left(\left\{\sup _{\omega \in S_{+,j n}}\left|T_{+,n}(\omega)-T_{+,n}\left(\tilde{l}_{+,\oplus}\right)\right| \geq C \frac{2^{2(j-1)}}{t_n^2}\right\} \cap B_{+,R}\right)
\end{aligned}
\end{equation}
for some $C>0$ where the last term goes to zero for any $\eta_2>0$ by Lemma \ref{lemma:l_hat-l_tilde}. Since
$$
d\left(\omega, \tilde{l}_{+,\oplus}\right)<\left(2^j / t_n\right)^{2 / \beta_2}
$$
on $S_{+,j n}$, this implies that the sum on the right-hand side of \eqref{eq:prob_rate_lhat} is bounded by
$$
4 b C^{-1} \sum_{\substack{j \geq M \\ 2^j \leq t_n \tilde{\eta}}} \frac{2^{2 j\left(1-\beta_2\right) / \beta_2}}{t_n^{2\left(1-\beta_2\right) / \beta_2} \sqrt{n h}} \leq 4 b C^{-1} \sum_{j \geq M}\left(\frac{1}{4^{\left(\beta_2-1\right) / \beta_2}}\right)^j.
$$
for some $b > 0$, using Markov's inequality and \eqref{eq:rate_delta_bound}. This converges for $M = M_n \to \infty$ since $\beta_2>1$. Hence,
$$
d\left(\hat{l}_{+,\oplus}, \tilde{l}_{+,\oplus}\right)=O_p\left(t_n^{2 / \beta_2}\right)=O_p\left[(n h)^{-\frac{1}{2\left(\beta_2-1\right)}}\right] .
$$ The result for $\hat{l}_{-,\oplus}$ follows analogously. 
\end{proof}
\end{lemma}

\begin{proposition} \label{prop:pooled_consistency}
Under Assumptions \ref{asspt:kernel}, 
\ref{asspt:bandwidth},
\ref{asspt:sampling}, 
\ref{asspt:densities},
\ref{asspt:uniqueness}-p,  \ref{asspt:curvature}-p, and \ref{assptT:J}, letting the optimal sequence $h=n^{-\gamma^*}$ with $\gamma^*=\left(\beta_1-1\right) /\left(4 \beta_2 + \beta_1-5\right)$
\[
d(\hat{l}_{p,\oplus}, m_{p,\oplus}) = O_p(n^{-2 /\left(\beta_1+4 \beta_2-5\right)}).
\]
\end{proposition}
\begin{proof}
First, I argue that $d\left(\hat{l}_{ p, \oplus}, \tilde{l}_{ p, \oplus}\right)=O_p\left[(n h)^{-1 /\left(2\left(\beta_2-1\right)\right)}\right]$, where 
$$\tilde{l}_{ p, \oplus} = \argmin_{\omega \in \Omega} E[s_{p,i}(h) d^2(Y, \omega)]$$ with 
$
s_{i,p}(h) \coloneqq \frac12 s_{+,i}(h) + \frac12 s_{-,i}(h).
$
Then, we essentially apply the proof of Lemma  \ref{lemma:aux_l_hat-l_tilde} twice for $s_{+,in}$ and $s_{-,in}$ and combine results. In particular, borrowing notation from that proof throughout, it follows immediately that $s_{p,in}(h)-s_{p,i}(h)=O_p\left((n h)^{-1 / 2}\right)$ and thus we can define,
\[
B_{p, R}:=\left\{\sup _{d\left(\omega, \tilde{l}_{p, \oplus}\right)<\delta}\left|\frac{1}{n} \sum_{i=1}^n\left[s_{p, i n}(h)-s_{p, i}(h)\right] D_{p, i}(\omega)\right| \leq R \delta(n h)^{-1 / 2}\right\}
\]
from some $R>0$ with $D_{p, i}(\omega) \coloneqq d^2(Y_i, \omega) - d^2(Y_i, \tilde{l}_{p,\oplus})$ so that $P(B^c_{p,R}) \to 0$. Similarly, we can define the pooled functions,
$
g_{p, \omega}(z, y) \coloneqq \frac12 g_{+, \omega}(z, y) + \frac12 g_{-, \omega}(z, y),
$
for which we can similarly construct an envelope $G_{p, n \delta}(z) = \frac12 G_{+, n \delta}(z) + \frac12 G_{-, n \delta}(z)$. Then we can use Theorems 2.7.11 and 2.14.2 of \citet{vaart1996weak} and \ref{assptT:J} to obtain, for small $\delta>0$, 
\[
E\left(\sup _{d\left(\omega, \tilde{l}_{p, \oplus}\right)<\delta}\left|\frac{1}{n} s_{p, i}(h) D_{p,i}(\omega)-E\left[s_{p,i}(h) D_{p,i}(\omega)\right]\right|\right)=O\left(\delta(n h)^{-1 / 2}\right).
\]
Then the rest of the proof for the stochastic term follows identically by using the modified Assumption \ref{asspt:curvature}-p instead of \ref{asspt:curvature}.
In a similar vein, one can easily show that $d\left(\tilde{l}_{p, \oplus}, m_{p, \oplus}\right)=O_p\left(h^{2 /\left(\beta_1-1\right)}\right)$ by applying the proof strategy of \ref{lemma:convergence_tilde} twice and combining to show that $\tilde{L}_{p,\oplus}(\omega) - M_{p,\oplus}(\omega) = O(h^2)$. The rest of the proof proceeds identically, but using the modified \ref{asspt:curvature}-p. Then the final result, $d(\hat{l}_{p,\oplus}, m_{p,\oplus}) = O_p((nh)^{-1/(2\beta_2 - 1)} + h^{2/(\beta_1-1)})$ follows by the triangle inequality. 
\end{proof}

Finally, the following auxiliary lemma is used in several proofs below,
\begin{lemma}\label{lem:s-avg-op} \label{lemma:l_bar-l_hat_op}
Under Assumptions \ref{asspt:kernel}, \ref{asspt:bandwidth},
\ref{asspt:sampling} and \ref{asspt:densities},
\[
  \frac1n\sum_{i=1}^n |s_{+,in}(h)| = O_p(1).
\]
\end{lemma}
\begin{proof}
From Lemma \ref{lemma:aux_l_hat-l_tilde} we have $\frac1n \sum_{i=1}^n |s_{+,in}(h)| = \frac1n \sum_{i=1}^n |s_{+,i}(h)| + O_p((nh)^{-1/2}$. Moreover, $E\left[|s_{+,i}(h)|\right] = O(1)$ and $E[s_{+,i}^2(h)]=O(h^{-1})$ by Lemma \ref{lemma:fan_gijbels} and \eqref{eq:E_k}, so the weak law of large numbers holds for $s_i$ since \begin{align*} 
& \var\left( \frac1n \sumn |s_{+,i}| \right) = \frac{1}{n^2} \sumn \var(|s_{+,i}|) = \frac{1}{n^2} \sumn \left(  E[|s_{+,i}^2|] - E[|s_{+,i}|]^2 \right) \\ 
& = \frac{1}{n^2} \sumn \left( O(h^{-1}) - O(1) \right) = O_p((nh)^{-1}).
\end{align*} As a result, $\frac1n \sum_{i=1}^n |s_{+,in}(h)| = O_p(1)$.
\end{proof}

\subsection{Proofs for Jump Test}
\label{app:anova_proofs}

% Denote the conditional Fr\'echet variance $V_{T,\oplus} = E[d^2(m_{T,\oplus}, Y_i) \mid X], \: T=0,1$.
%==========================

%==========================
\begin{lemma} \label{lemma:clt_first_term}
 Under Assumptions  \ref{asspt:sampling}-\ref{assptT:J},
 \[ 
 \frac1n \sum_{i=1}^n \left( s_{\pm, i n}(h) \left( d^2(\hat{l}_{\pm, \oplus}, Y_i) - d^2(m_{\pm, \oplus}, Y_i) \right) \right) = o_p\left((nh)^{-1/2}\right).
 \]
\begin{proof}
Write $\Delta=d(\hat l_{+,\oplus}, m_{+,\oplus})$ and 
$D_i(\omega)=d^{2}(Y_i,\omega)-d^{2}(Y_i,m_{+,\oplus})$. The term to be bounded is $\frac{1}{n} \sum_{i=1}^n s_{+,in}(h) D_i(\hat{l}_{+,\oplus})$.
By the reverse triangle inequality and the boundedness of the metric space $\Omega$, we have that $|D_i(\hat l_{+,\oplus})| \le 2\operatorname{diam}(\Omega)\,\Delta$. This implies,
\begin{equation}
\Bigl|\tfrac1n\sum_{i=1}^n s_{+,in}(h)\,D_i(\hat l_{+,\oplus})\Bigr|
      \le 2\operatorname{diam}(\Omega)\,\Delta\,
           \left( \frac1n \sum_{i=1}^n |s_{+,in}(h)| \right).
 \label{eq:plugin_bound}
\end{equation}
From Lemma~\ref{lem:s-avg-op}, we have, $\tfrac1n\sum_i |s_{+,in}(h)|=O_p(1)$. 
Therefore, the entire expression in \eqref{eq:plugin_bound} is of order $O_p(\Delta)$. The proof thus reduces to showing that the first-stage estimation error $\Delta$ converges to zero faster than the target rate of the second-stage estimator,  $\Delta = o_p((nh)^{-1/2})$.

By the triangle inequality and the convergence rates established in Lemmas~\ref{lemma:convergence_tilde} and \ref{lemma:aux_l_hat-l_tilde}, we have,
\[
\Delta \le \underbrace{d(\tilde{l}_{+,\oplus}, m_{+,\oplus})}_{\text{Bias Term}} + \underbrace{d(\hat{l}_{+,\oplus}, \tilde{l}_{+,\oplus})}_{\text{Stochastic Term}} = O_p(h_m^{2/(\beta_1-1)}) + O_p((nh_m)^{-1/(2(\beta_2-1))}).
\]
Letting $h_m = n^{-\theta}$ and $h = n^{-\gamma}$, the target rate is $(nh)^{-1/2} = n^{-(1-\gamma)/2}$. I analyze each component to show it is of a smaller order than this target rate.

First, for the bias component of $\Delta$, I need $h_m^{2/(\beta_1-1)} = o_p(n^{-(1-\gamma)/2})$, which is equivalent to $1 - \gamma - \frac{4\theta}{\beta_1-1} < 0$. This is a relative rate condition on the bandwidths, enforced by requiring $nh h_m^{4/(\beta_1-1)} \to 0$ in Assumption~\ref{asspt:bandwidth}(a).

Second, for the stochastic component of $\Delta$, I require $(nh_m)^{-1/(2(\beta_2-1))} = o_p((nh)^{-1/2})$, which is equivalent to $\frac{1-\theta}{2(\beta_2-1)} > \frac{1-\gamma}{2}$ in Assumption~\ref{asspt:bandwidth}(d).

Since both the bias and stochastic components of $\Delta$ are of a strictly smaller order than $(nh)^{-1/2}$ under these conditions, it follows that $\Delta = o_p((nh)^{-1/2})$. Substituting this result back into \eqref{eq:plugin_bound} shows that the entire plug-in term is $o_p((nh)^{-1/2})$, which completes the proof.
\end{proof}
\end{lemma}

The next few results prepare the grounds for proving the central limit theorem for conditional Fr\'echet variances. For full generality, I first prove it without the undersmoothing condition, so that an asymptotic bias term of order $h^2$ shows up, which cancels out when imposing the undersmoothing. For that purpose, I need to define the following bias and variance components. Let $V_{\pm,\oplus}''(c)$ denote the second derivative of the Fr\'echet variance from the left/right, evaluated at $x=c$.  For $K_{+,kj} = \int_0^\infty u^j K^k(u) \dd u$, $K_{-,kj} = \int_{-\infty}^0 u^j K^k(u) \dd u$, remember that,
\begin{equation} \label{eq:S_plus_main}
S_\pm \coloneqq \frac{\int_{0}^\infty (K_{\pm,12} - u K_{\pm,11})^2 K^2(u) \dd u}{ (K_{\pm,12} K_{\pm,10} - K^2_{\pm,11})^2 },
\end{equation}
and define,
\begin{equation} \label{eq:B_plus_main}
B_\pm \coloneqq \frac{(K_{\pm,12}^2 - K_{+,11} K_{\pm,13})}{K_{\pm,12} K_{\pm,10} - K_{\pm,11}^2}.
\end{equation}
$S_\pm$ and $B_\pm$ are defined analogously with $K_{\pm,kj} = \int_{\pm\infty}^0 u^j K^k(u) \dd u$. Note that $B_+ = B_-$ and $S_+ = S_-$ here due to the assumption of a symmetric kernel, but I use this notation to align with standard notation in local polynomial regression \citep{fan1992variable}.

The following result restates the well-known expression of the bias of local polynomial regression estimators from \citet[Theorem 4]{fan1992variable} (see also the general derivations in \citet[Lemma A.1(B)]{calonico2014robust}).
\begin{lemma} \label{lemma:clt_bias} 
Let $m(x)$ be a function with bounded second derivative. Under Assumptions \ref{asspt:kernel}, \ref{asspt:bandwidth}, \ref{asspt:sampling}, \ref{asspt:densities},
\begin{align*}
& \frac1n \sum s_{\pm, in}(h) \left( m(X_i) - m(c) \right) = \frac{h^2}{2} m''(c) B_{\pm} (1 + o_p(1)) 
\end{align*}
where
$$
B_{\pm}:=\frac{\left(K_{\pm, 12}^2-K_{\pm, 11} K_{\pm, 13}\right)}{K_{\pm, 12} K_{\pm, 10}-K_{\pm, 11}^2} .
$$
\begin{proof}
The equality follows from a well-known expression for the bias term in local polynomial regression and follows from the proof of Theorem 4 in \citet{fan1992variable}, with $s_{\pm, in}$ corresponding to $w_j / \sum_{j=1}^n w_j$ in their notation, since $m(\cdot)$ is bounded and differentiable with respect to its argument $x$. 
\end{proof}
\end{lemma}

%================
 \begin{lemma} \label{lemma:lyapunov}
 Define
$$
Z^{(k)}_{\pm, in} \coloneqq  s_{\pm, in}(h) d^k(m_{\pm,\oplus}, Y_i) - E\left[ s_{\pm, in}(h) d^k(m_{\pm,\oplus}, Y_i) \mid \mathbf{X} \right],
$$
for $k > 0$ where $\mathbf{X} \coloneqq \{X_i\}_{i=1}^n$ and $Y_i$ is in a bounded metric space $(\Omega, d)$. Assume \ref{asspt:kernel}, \ref{asspt:bandwidth}, \ref{asspt:sampling}, \ref{asspt:densities}, \ref{asspt:uniqueness}. Then,
$$
\lim_{n \to \infty} P\left( \frac{\sumn Z^{(k)}_{\pm, in}}{S^{(k)}_{\pm,n}}  \leq t \mid \mathbf{X} \right) = \Phi(0,1),
$$
for all $t \in \mathbb{R}$, with $\Phi(0,1)$ the normal CDF and 
\[ 
(S^{(k)}_{\pm, n})^2 \coloneqq \sumn \operatorname{var}(Z^{(k)}_{\pm, in} \mid \mathbf{X}) =  \left(  \frac{S_\pm }{f_X(c) nh } \var(d^k(m_{\pm, \oplus}, Y) \mid X) \right) (1 + o_p(1)),
\]
where $S_\pm$ were defined above.
\begin{proof}
I will use Lyapunov's central limit theorem \citep[Theorem 27.3]{billingsley2017probability} to prove the result, as, given $\mathbf{X}$, $s_{+, in}(h) d^k(m_{+,\oplus}, Y_i)$ are independent but not identically distributed. To that end, I need to verify Lyapunov's condition,
\begin{equation} \label{eq:lyapunov_condition}
\lim _{n \rightarrow \infty} \frac{1}{(S^{(k)}_{+,n})^{2+\delta}} \sum_{i=1}^n E\left[\left|Z^{(k)}_{+,in}\right|^{2+\delta} \mid \mathbf{X}\right]=0,
\end{equation}
for some $\delta >0$. First, note that the boundedness of $\Omega$ and the existence of $s_{+,in}(h)$ for finite $n$ imply that all conditional moments $
E[|Z^{(k)}_{+,in}|^\ell \mid \mathbf{X} ]$ exist for finite $n$ and $\ell>0$.
Then, by the results in Lemma \ref{lemma:fan_gijbels}, $\mu_{+,j}=O(h^j), \sigma_{+,0}^2 = O(h^2)$ and thus
\begin{align*}
| s_{+,i} |  & = | \mathrm{1}(X_i \geq c) \sigma_{+,0}^{-2} K_h(X_i-c) \left( \mu_2 - \mu_1 (X_i-c) \right) | \\ 
& = |  \mathrm{1}(X_i \geq c) \sigma_{+,0}^{-2} K\left(\frac{X_i-c}{h}\right) \left( \frac{\mu_2}{h} - \mu_1 \frac{(X_i-c)}{h} \right) | \\ 
& \leq |  \sigma_{+,0}^{-2} \bar{K}  \frac{\mu_2}{h} | +  | \sigma_{+,0}^{-2} \bar{K} \mu_1 M | \\
& =  O(h^{-2}) O(h) + O(h^{-2}) O(h) \\
& = O(h^{-1})
\end{align*}
where I denote $\bar{K}$ the kernel's upper bound and $M$ the upper bound of its support. The inequality follows from the triangle inequality, the boundedness of the kernel, and the fact that the kernel is only non-zero when $\frac{(X_i-c)}{h} \leq M$ since it has compact support. As a result, $|s_{+,in}|^k = |s_{+,i} + O_p((nh)^{-1/2})|^k = |O(h^{-1}) + O_p((nh)^{-1/2})|^k = O_p(h^{-k})$ using previous results from \ref{lemma:l_hat-l_tilde}. 
%by the arguments in the previous lemmas, $s_{+, in} - s_{+,i} = O_p((nh)^{-1/2})$ and $E[|s_{+,i}|] = O(1)$ so $s_{+, in} = O_p(1)$ % 
 Then, observe that the numerator in \eqref{eq:lyapunov_condition} equals,
 \begin{align*}
 & \sum_{i=1}^n\left|s_{+, i n}(h)\right|^{2+\delta} E\left[ \left|d^k\left(Y_i, m_{+, \oplus}\right)-E\left[d^k\left(Y_i, m_{+, \oplus}\right) \mid \mathbf{X}\right]\right|^{2+\delta} \mid \mathbf{X} \right] \\ 
 & \leq n O_p(h^{-(2+\delta)}) 2 \mathrm{diam}(\Omega)^{k(2+\delta)} = O_p(nh^{-(2+\delta)}).
 \end{align*}
 where $s_{+,in}(h)$ is a constant conditional on $\mathbf{X}$. Similarly, 
 \begin{align*}
\var(Z^{(k)}_{+,in} \mid \mathbf{X}) \leq s_{+,in}(h)^2 4 \mathrm{diam}(\Omega)^{2k} = O_p(h^{-2}).
 \end{align*}
 As a result, the denominator equals,
 \begin{align*}
 \left(\sumn \operatorname{var}(Z^{(k)}_{\pm, in} \mid \mathbf{X})\right)^{1+\delta/2} & = O_p(n^{1+\delta/2} h^{2+\delta}).
 \end{align*}
 Thus, 
 \begin{align*}
\eqref{eq:lyapunov_condition} = \frac{O_p(nh^{-(2+\delta))}}{O_p(n^{1+\delta/2} h^{2+\delta}}) = O_p(n^{-\delta})
 \end{align*}
 which goes to zero, proving the result. The result for $Z^{(k)}_{-,in}$ obtains analogously. 

 Finally, the expression for $(S^{(k)}_{\pm,n})^2$ is the standard expression for the variance term in \citet[Theorem 4]{fan1992variable}, but replacing $\var(Y \mid X=c)$ with $\var(d^k(m_{+,\oplus}, Y) \mid X=c)$. To see this, note that \[\sum_{i=1}^n \operatorname{var}\left(Z_{+, \text {in }}^{(k)} \mid \mathbf{X}\right) = \sumn s_{+,in}(h)^2 \operatorname{var}\left(Z_{+, \text {in }}^{(k)} \mid X\right)\] by the fact that $X$ is i.i.d, which is exactly the form of the variance expression $\frac{\Sigma_1^n w_j^2 \sigma^2\left(X_j\right)}{\left(\Sigma_1^n w_j\right)^2}$ in the notation of \citet{fan1992variable}. The final expression then follows from a Taylor expansion around $c$ and the properties of the kernel. Derivations for $Z^{(k)}_{-,in}$ are analogous. 
\end{proof}
 \end{lemma}

 %====================

\begin{theorem} \label{thm:clt_variance} Under Assumptions \ref{asspt:sampling}--\ref{assptT:J},
\[
\sqrt{nh} \left( \hat{V}_{\pm, \oplus} - V_{\pm, \oplus}  \right)  \xrightarrow[]{d} N\left(0, \frac{S_\pm}{f_X(c)} \sigma^2_{+, V}\right).
\]
\begin{proof}
Decomposing the estimation error, we have
\begin{align*}
\sqrt{nh} \left( \hat{V}_{+, \oplus} - V_{+, \oplus} \right) & = \sqrt{nh} \left(\frac1n \sum_{i=1}^n s_{+,in}(h) d^2(\hat{l}_{+, \oplus}, Y_i) - \lim_{x \to c^+} E\left[ d^2(m_{+,\oplus}, Y) \mid X=x \right] \right) \\ 
& = \sqrt{nh} \left( A_{+,1} + A_{+,2}  + A_{+,3} \right),
\end{align*}
where remember that $V_{+, \oplus} = \lim_{x \to c^+} E\left[d^2(m_{\oplus}(x), Y) \mid X=x \right]$ and the three terms are,
\begin{align*}
A_{+,1} & \coloneqq \frac1n \sum_{i=1}^n s_{+,in}(h) \left( d^2(\hat{l}_{+, \oplus}, Y_i) - d^2(m_{+,\oplus}, Y_i)\right) \\ 
A_{+,2} & \coloneqq \frac1n \sum_{i=1}^n s_{+,in}(h) \left( d^2(m_{+,\oplus}, Y_i)    - E\left[ d^2(m_{+,\oplus}, Y_i) \mid X_i \right]  \right ) \\ 
A_{+,3} & \coloneqq \frac1n \sumn s_{+,in}(h) \left( E\left[ d^2(m_{+,\oplus}, Y) \mid X_i \right] - \lim_{x \to c^+} E\left[d^2(m_{+,\oplus}, Y) \mid X=x \right] \right).
\end{align*}
I analyze each term in turn.
\begin{itemize}
\item \textit{$A_{+,1}$ (Plug-in Error):} This term accounts for the error from using the estimated Fr\'echet mean $\hat{l}_{+, \oplus}$. Under the bandwidth conditions in Assumption \ref{asspt:bandwidth}, Lemma \ref{lemma:clt_first_term} shows this term is asymptotically negligible, with $\sqrt{nh} A_{+,1} = o_p(1)$.

\item \textit{$A_{+,2}$ (Stochastic Term):} This is the main stochastic component. Conditional on $\mathbf{X} = \{X_i\}_{i=1}^n$, the terms in the sum are independent. By the conditional central limit theorem in Lemma \ref{lemma:lyapunov},
$$
\sqrt{nh} A_{+,2} \xrightarrow{d} N\left(0, \frac{S_+}{f_X(c)} \sigma^2_{+,V}\right).
$$
\item \textit{$A_{+,3}$ (Bias Term):} This term represents the asymptotic bias. Let $g_+(x) \coloneqq E\left[ d^2(m_{+,\oplus}, Y) \mid X=x \right]$ for $x \ge c$. Assumption \ref{asspt:densities} ensures $g_+(x)$ is twice continuously differentiable in a right-sided neighborhood of $c$. The expression for $A_{+,3}$ is the standard form for the bias of a local linear estimator at a boundary. By Lemma \ref{lemma:clt_bias} and the undersmoothing condition $nh^5 \to 0$ from Assumption \ref{asspt:bandwidth}(c), the scaled bias vanishes: $\sqrt{nh} A_{+,3} \to 0$.
\end{itemize}
Combining the results for the three terms via Slutsky's lemma yields the result. 
\end{proof}
\end{theorem}
%====================

\begin{proposition} \label{prop:sigma_v_convergence}
Assume \ref{asspt:sampling}--\ref{assptT:J}. Then
 \[
 \hat{\sigma}_{\pm,V} -  \sigma_{\pm,V} = O_p((nh)^{-1/2}).
 \]
\begin{proof}
The proof proceeds by establishing a multivariate central limit theorem for the estimators of the second and fourth moments, and then applying the Delta Method.

First, we handle the plug-in error from using $\hat{l}_{+, \oplus}$ instead of $m_{+, \oplus}$. By the triangle inequality,
\begin{align*}
\left| \frac{1}{n} \sum_{i=1}^n s_{+,in}(h) d^4\left(\hat{l}_{+,\oplus}, Y_i\right) - \frac{1}{n} \sum_{i=1}^n s_{+,in}(h) d^4\left(m_{+,\oplus}, Y_i\right) \right| \leq  \\
2 \operatorname{diam}^2(\Omega) \left| \frac{1}{n} \sum_{i=1}^n s_{+,in}(h) \left(d^2\left(\hat{l}_{+,\oplus}, Y_i\right) - d^2\left(m_{+,\oplus}, Y_i\right)\right) \right|
\end{align*}
where the right-hand side is $o_p((nh)^{-1/2})$ by the same logic as in Lemma \ref{lemma:clt_first_term}. Thus, the error from the first-stage mean estimation is asymptotically negligible, and we can analyze the vector of moment estimators as if they were constructed using the true mean $m_{+,\oplus}$.

The target is the vector of true moments at the limit:
$$
\mathbf{M} \coloneqq \binom{\lim_{x \to c^+} E\left[d^4\left(m_{+,\oplus}, Y\right) \mid X=x\right]}{\lim_{x \to c^+} E\left[d^2\left(m_{+,\oplus}, Y\right) \mid X=x\right]}.
$$
The error vector of the estimator can be decomposed into a stochastic term and a bias term:
\begin{align*}
& \sqrt{nh} \left[\binom{\frac{1}{n} \sum_{i=1}^n s_{+,in}(h) d^4\left(m_{+,\oplus}, Y_i\right)}{\frac{1}{n} \sum_{i=1}^n s_{+,in}(h) d^2\left(m_{+,\oplus}, Y_i\right)} - \mathbf{M}\right] \\
& = \sqrt{nh} \left[\binom{\frac{1}{n} \sum s_{+,in}(h) \left( d^4(m_{+,\oplus}, Y_i) - E[d^4(m_{+,\oplus}, Y_i)|X_i]\right)}{\frac{1}{n} \sum s_{+,in}(h) \left( d^2(m_{+,\oplus}, Y_i) - E[d^2(m_{+,\oplus}, Y_i)|X_i]\right)}\right] \quad \text{(Stochastic Term)} \\
& + \sqrt{nh} \left[\binom{\frac{1}{n} \sum s_{+,in}(h) E[d^4(m_{+,\oplus}, Y_i)|X_i] - \lim_{x \to c^+} E[d^4(m_{+,\oplus}, Y)|X=x]}{\frac{1}{n} \sum s_{+,in}(h) E[d^2(m_{+,\oplus}, Y_i)|X_i] - \lim_{x \to c^+} E[d^2(m_{+,\oplus}, Y)|X=x]}\right]. \quad \text{(Bias Term)}
\end{align*}
The bias term vanishes as $n \to \infty$ by an application of Lemma \ref{lemma:clt_bias} to each component, combined with the undersmoothing condition in Assumption \ref{asspt:bandwidth}.

The stochastic term converges in distribution to a multivariate normal. This is established via the Cram\'er-Wold device \citep[Theorem 29.4]{billingsley2017probability}. Let $Z^{\mathrm{lin}}_{+,in} \coloneqq t_4 Z^{(4)}_{+,in} + t_2 Z^{(2)}_{+,in}$ for a fixed vector $\mathbf{t} = (t_4, t_2)^\top \in \mathbb{R}^2$. I verify that the conditional Lyapunov condition holds in probability.

The numerator of the Lyapunov ratio is bounded as follows,
\begin{align*}
\sum_{i=1}^n E\left[|Z^{\mathrm{lin}}_{+,in}|^{2+\delta} \mid \mathbf{X}\right] &= \sum_{i=1}^n |s_{+,in}(h)|^{2+\delta} E\left[\left| \sum_{k \in \{2,4\}} t_k\left(d^k(\dots) - E[d^k|X_i]\right) \right|^{2+\delta} \mid X_i \right] \\
&\leq \sum_{i=1}^n |s_{+,in}(h)|^{2+\delta} \cdot C_\delta \sum_{k \in \{2,4\}} |t_k|^{2+\delta} E\left[\left| d^k(\dots) - E[d^k|X_i] \right|^{2+\delta} \mid X_i \right] \\
&\leq \sum_{i=1}^n |s_{+,in}(h)|^{2+\delta} \cdot C_\delta \sum_{k \in \{2,4\}} |t_k|^{2+\delta} \left(2\operatorname{diam}(\Omega)^k\right)^{2+\delta} \\
&= \sum_{i=1}^n O_p(h^{-(2+\delta)}) \cdot C_1(\mathbf{t}) = O_p(n h^{-(2+\delta)}),
\end{align*}
where the second line uses the $c_r$-inequality and the third line uses the boundedness of the metric space. $C_1(\mathbf{t})$ is a constant that depends on $\mathbf{t}$ and $\operatorname{diam}(\Omega)$.

The denominator is the sum of variances, which can be bounded by,
\begin{equation*}
\left(S_{+, n}^{\text {lin }}\right)^2=\sum_{i=1}^n \operatorname{var}\left(Z^{\text {lin }}_{+,in} \mid \mathbf{X}\right) \leq \sum_{i=1}^n s_{+, in }(h)^2 \sum_{k=2,4} 4 t_k \operatorname{diam}(\Omega)^{2 k}=O_p\left(n h^{-2}\right)
\end{equation*}
The full denominator of the Lyapunov condition is this sum raised to the power $1+\delta/2$, which is of order $O_p\left( (nh^{-2})^{1+\delta/2} \right) = O_p\left( n^{1+\delta/2} h^{-(2+\delta)} \right)$. The ratio of the orders of the numerator and denominator is $n^{-\delta/2}$, which converges to 0. Thus, the Lyapunov condition holds in probability.  

Since the Lyapunov condition holds for any arbitrary, fixed linear combination $(t_4, t_2)$, the sum of the linearly combined variables converges to a normal distribution. Further, \[
\left(S_{+, n}^{\operatorname{lin}}\right)^2=\left(\frac{S_{+}}{f_x(c) n h} \operatorname{var}\left(\sum_{k=2,4} t_k d^k\left(m_{+, \oplus}, Y\right) \mid X\right)\right)\left(1+o_p(1)\right)\] again using the derivations in Theorem 4 of \citet{fan1992variable} and the linearity of conditional expectations. Thus, we can use the Cram\'er-Wold device, which implies that the original vector of stochastic terms converges to a multivariate normal distribution,
$$
\sqrt{nh} \left[\binom{\frac{1}{n} \sum_{i=1}^n s_{+,in}(h) d^4\left(m_{+,\oplus}, Y_i\right)}{\frac{1}{n} \sum_{i=1}^n s_{+,in}(h) d^2\left(m_{+,\oplus}, Y_i\right)} - \mathbf{M}\right] \xrightarrow{d} N(0, \Sigma),
$$
for some asymptotic covariance matrix $\Sigma$.

Finally, observe that $\hat{\sigma}_{+,V}^2$ is a continuous function of the estimated moments:
$$
\hat{\sigma}_{+,V}^2=g\left(\frac{1}{n} \sum s_{+,in} d^4\left(\hat{l}_{+,\oplus}, Y_i\right), \frac{1}{n} \sum s_{+,in} d^2\left(\hat{l}_{+,\oplus}, Y_i\right)\right),
$$
where $g(x_1, x_2) = x_1 - x_2^2$ is continuously differentiable with gradient $\nabla g(x_1, x_2)=(1, -2x_2)$.
An application of the multivariate Delta Method to the CLT result yields,
$$
\sqrt{nh} \left( \hat{\sigma}_{+,V}^2 - \sigma_{+,V}^2 \right) \xrightarrow{d} N\left(0, \nabla g(\mathbf{M})^\top \Sigma \nabla g(\mathbf{M})\right).
$$
This implies the desired result,
\[
\hat{\sigma}_{+,V} - \sigma_{+,V} = O_p((nh)^{-1/2}).
\]
The proof for $\hat{\sigma}_{-,V}$ is analogous.
\end{proof}
\end{proposition}

\paragraph{Proof of Theorem \ref{thm:CLT}.} 

\begin{proof}
 From Theorem \ref{thm:clt_variance} and Proposition \ref{prop:sigma_v_convergence}, the result follows by applying Slutsky's lemma and the continuous mapping theorem. Results for $\hat{V}_{-, \oplus}$ obtain analogously. 
\end{proof}

\paragraph{Proof of Proposition \ref{prop:F_n}.}
\begin{proof}
Under the null hypothesis the conditional Fr\'echet means ``to the left'' and ``to the right'' are equal to the pooled and standard conditional Fr\'echet mean,
$$
m_{-,\oplus}=m_{+, \oplus}=m_{p,\oplus}=m_{\oplus}(c).
$$
As a result, under Assumptions \ref{asspt:uniqueness}, \ref{assptT:J},
$$
\begin{aligned}
& \sqrt{n h} F_n  \\
= & \sqrt{nh} \left( \hat{V}_{p, \oplus} - \left( \frac12 \hat{V}_{+,\oplus} + \frac12 \hat{V}_{-,\oplus} \right) \right) \\
 = &  \sqrt{nh} \left(\frac1n\sum s_{p, in}(h) d^2(Y_i, \hat{l}_{p, \oplus}) -  \frac{1}{2n} \sum_{j=+,-} \sumn s_{j,in}(h) d^2(Y_i, \hat{l}_{j,\oplus}) \right)  \\ 
= &  \sqrt{nh} \frac1n \sum_{i=1}^n s_{p, in}(h) \left\{ d^2(\hat{l}_{\oplus}, Y_i)-  d^2\left(m_{p,\oplus}, Y_i\right)\right\} \\ 
& - \sqrt{nh} \frac{1}{2n} \sumn s_{+,in}(h) \left\{d^2\left(\hat{l}_{+, \oplus}, Y_i\right)-d^2\left(m_{+,\oplus}, Y_i\right)\right\}  \\
& -  \sqrt{nh}  \frac{1}{2n} \sumn s_{+,in}(h) \left\{d^2\left(\hat{l}_{-, \oplus}, Y_i\right)-d^2\left(m_{-,\oplus}, Y_i\right)\right\}    \\
& =   o_p(1) 
\end{aligned}
$$
where the third equality follows from the null hypothesis and the last equality follows by applying Lemma \ref{lemma:clt_first_term} to the three terms, noting that the proof of that Lemma holds equally for $\hat{l}_{p,\oplus}$ by substituting the bounds derived in Proposition \ref{prop:pooled_consistency}. 
\end{proof}

 \paragraph{Proof of Proposition \ref{prop:U_n}.}
\begin{proof}
Under the null hypothesis of equal Fr\'echet variances, we have $V_{+,\oplus} = V_{-,\oplus} = V_\oplus$. Define the deviations
\[
\tilde{V}_+ = \hat{V}_{+, \oplus} - V_\oplus, \quad \tilde{V}_- = \hat{V}_{-, \oplus} - V_\oplus.
\]
Then, note that $S_+ = S_-$ by our Assumption that the kernel is symmetric in \ref{asspt:kernel}. 
Then, Theorem \ref{thm:CLT}, Proposition \ref{prop:sigma_v_convergence}, the assumption that $\hat{f}_X(c)$ is $\sqrt{nh}$-consistent and Slutsky's lemma imply, under the null hypothesis,
\[
\frac{\sqrt{nh} \left( \tilde{V}_+ - \tilde{V}_- \right) }{ \sqrt{(\hat{S}_+ \hat{\sigma}_{+, V}^2 + \hat{S}_- \hat{\sigma}_{-, V}^2) / \hat{f}_X(c)}  }\xrightarrow{d} N\left( 0, 1 \right).
\]
Further, note that under the null hypothesis,
\[
nh U_n = nh \frac{ \left( \hat{V}_{+, \oplus} - \hat{V}_{-, \oplus} \right)^2 }{ (S_+ \hat{\sigma}_{+, V}^2 + S_- \hat{\sigma}_{-, V}^2) / \hat{f}_X(c) } = \frac{ \left( \sqrt{nh} \left( \tilde{V}_+ - \tilde{V}_- \right) \right)^2 }{( S_+ \hat{\sigma}_{+, V}^2  + S_- \hat{\sigma}_{-, V}^2) / \hat{f}_X(c) }.
\]
Applying the continuous mapping theorem to the previous results, it follows that,
\[
nh U_n \xrightarrow{d} \chi_1^2.
\]
\end{proof}

The following lemma is explicitly proved for convenience but is standard and follows from previous results.
\begin{lemma} \label{lemma:s_in_to_E}
Under Assumptions \ref{asspt:kernel} and \ref{asspt:densities}, for any $\omega \in \Omega$, 
\[
\frac1n \sumn s_{\pm,in} d^2(Y_i, \omega)  - \lim_{x \to c^\pm} E[ d^2(Y_i, \omega) | X=x] = o_p(1)
\]
\begin{proof}
The absolute value of the left-hand side is less than or equal to,
\begin{align*}
 &\left| \frac1n \sumn (s_{+,in}(h)- s_{+,i}(h)) d^2(Y_i, \omega) \right| + \\
& \left| \frac1n \sumn s_{+,i}(h) d^2(Y_i, \omega) -  E[s_{+,i}(h) d^2(Y_i, \omega)] \right| + \\
& \left| \frac1n \sumn E[s_{+,i}(h) d^2(Y_i, \omega)] - \lim_{x \to c^\pm} E[ d^2(Y_i, \omega) | X=x] \right| = o_p(1)
\end{align*}
where the first term is $o_p(1)$ as $s_{+,in}(h) - s_{+,i}(h)=O_p((nh)^{-1/2})$ as shown in the proof of Lemma \ref{lemma:aux_l_hat-l_tilde}; the second term is $O_p((nh)^{-1/2})$ by the weak law of large numbers since $\var(s_{+,i} d^2(Y_i,\omega)) \to 0$ by the same argument as in the proof Lemma \ref{lemma:l_bar-l_hat_op} and the boundedness of $\Omega$; and the last term is $O(h^2)$ by the results in the proof of Lemma \ref{lemma:convergence_tilde}.
\end{proof}
\end{lemma}

\paragraph{Proof of Proposition \ref{prop:F_n_F_consistency}}
\begin{proof}
The consistency of $\hat{V}_{\pm, \oplus}$ is implied by Theorem \ref{thm:CLT}. The consistency of $\hat{V}_{p,\oplus}$ follows by noting that, 
\begin{align*}
\left|\hat{V}_{p,\oplus}-\frac{1}{n} \sum_{i=1}^n s_{p,in}(h)d^2\left(m_{p,\oplus}, Y_i\right)\right| \leq 2 \operatorname{diam}(\Omega) d\left(\hat{l}_{p,\oplus}, m_{p,\oplus}\right) \frac1n \sumn |s_{p,in}(h)|=o_p(1),
\end{align*}
which follows from the consistency of the pooled means in Proposition \ref{prop:pooled_consistency} and the fact that $\frac1n \sumn |s_{\pm,in}|=O(1)$ as shown in the proof of Lemma \ref{lemma:l_bar-l_hat_op}. Together with Lemma \ref{lemma:s_in_to_E}, applied twice, that implies,
\begin{align*}
 \left|\hat{V}_{p,\oplus} - V_{p,\oplus} \right| 
 & \leq \left| \frac1n \sumn s_{p,in}(h)d^2(Y_i, \hat{l}_{p,\oplus}) - \frac1n \sumn s_{p,in}(h)d^2(Y_i, m_{p,\oplus}) \right| \\ 
 &  + \left| \frac1n \sumn s_{p,in}(h)d^2(Y_i, m_{p,\oplus})  - V_{p,\oplus} \right| = o_p(1).
\end{align*}
Then,
\begin{align*}
 V_{p,\oplus} - \frac12 \left( V_{+,\oplus} +  V_{-,\oplus}\right)  &  = \frac12 \lim_{x \to c^+} E\left[ d^2(Y, m_{p,\oplus}) - d^2(Y, m_{+,\oplus}) \mid X=x \right] \\
 & \quad + \frac12 \lim_{x \to c^-} E\left[ d^2(Y, m_{p,\oplus}) - d^2(Y, m_{-,\oplus}) \mid X=x \right].
\end{align*}
By the definition of the one-sided Fr\'echet means $m_{+,\oplus}$ and $m_{-,\oplus}$ as the unique minimizers of their respective expected squared distance functions (see Eq. \eqref{eq:frechet_cond_mean_pm}), both terms on the right-hand side are non-negative. Therefore, $F_{pop} \ge 0$. Equality to zero occurs if and only if $m_{p,\oplus}$ is also the minimizer for each one-sided expectation, which by the uniqueness assumption (\ref{asspt:uniqueness}) requires $m_{p,\oplus}=m_{+,\oplus}$ and $m_{p,\oplus}=m_{-,\oplus}$. As a result, $F_{pop}=0$ if and only if the population means are all equal. The properties of $U_{pop}$ follow directly from its definition.
\end{proof}

\paragraph{Proof of Theorem \ref{thm:consistency}}
\begin{proof}
The proof is similar to the proof of Theorem 3 in \citet{dubey2019frechet}. Suppressing the subscripts $\oplus$ to ease notation, we first need the following lemma, which is proved further below.
\begin{lemma} \label{lemma:consistency_sup_p}
 Under the assumptions of Theorem \ref{thm:consistency}, it holds for all $\epsilon>0$:
\begin{enumerate}[(a)]
\item $\lim _{n \rightarrow \infty}\left\{\sup _{P \in \mathcal{P}} P(|\hat{V}_{+}-V_{+}|>\epsilon)\right\}=0$;
\item $\lim _{n \rightarrow \infty}\left\{\sup _{P \in \mathcal{P}} P\left(\left|\hat{\sigma}_{+,V}^2-\sigma_{+,V}^2\right|>\epsilon\right)\right\}=0$;
\item $\lim _{n \rightarrow \infty}\left\{\sup _{P \in \mathcal{P}} P\left(\left|\hat{V}_{p}-V_{p}\right|>\epsilon\right)\right\}=0$,
\end{enumerate}
where (a) and (b) hold analogously for the `-' estimators. The supremum is taken with respect to the underlying true probability measure $P$ of $\{Y_i, X_i \}_{i=1}^n$, over the class $\mathcal{P}$ of possible probability measures which generate random observations from $(\Omega, \mathbb{R})$.
\end{lemma}
Then, define as before,
$$
\tilde{V}_{+}=\hat{V}_{+}-V_{+}
$$
and similarly for $\hat{V}_-$. The statistic $U_n$ can then be written as 
$$
U_n=  \frac{\left(\hat{V}_{+}-\hat{V}_{-}\right)^2}{ \left( \hat{\sigma}^2_{+,U} + \hat{\sigma}^2_{-,U} \right)} = \tilde{U}_n+\Delta_n,
$$
where $\hat{\sigma}^2_{\pm,U} \coloneqq S_{\pm} \hat{\sigma}_{\pm, V}^2 / \hat{f}_X(c)$ and
$$
\tilde{U}_n=\frac{\left(\tilde{V}_+-\tilde{V}_-\right)^2}{\hat{\sigma}^2_{+,U} + \hat{\sigma}^2_{-,U}}
$$
and
$$
\Delta_n=\frac{  \left(V_+-V_-\right)^2+2 \left(V_+-V_-\right)(\tilde{V}_+-\tilde{V}_-) }{\hat{\sigma}^2_{+,U} + \hat{\sigma}^2_{-,U}} 
$$
Repeating the proof of Proposition \ref{prop:U_n}, we obtain
$$
nh \tilde{U}_n \rightarrow \chi_{1}^2 \quad \text { in distribution. }
$$
Further, by Lemma \ref{lemma:consistency_sup_p} and continuity, 
$$
\Delta_n +\frac{F_n^2}{\hat{\sigma}^2_{+,U} + \hat{\sigma}^2_{-,U}}
$$
is a uniformly consistent estimator of
$$
U+\frac{F^2}{\sigma^2_{+,U} + \sigma^2_{-,U}},
$$
with $\sigma^2_{\pm, U} \coloneqq S_{ \pm} \sigma_{ \pm, V}^2 / \hat{f}_X(c)$. 

That implies, for sets $\left\{A_n\right\}$, defined as, 
$$
A_n=\left\{\Delta_n +\frac{F_n^2}{\hat{\sigma}^2_{+,U} + \hat{\sigma}^2_{-,U}} < \frac12 \left( U+\frac{F^2}{\sigma^2_{+,U} + \sigma^2_{-,U}}\right) \right\},
$$
that as $n\to \infty$, by Lemma \ref{lemma:consistency_sup_p}
$$
\begin{aligned}
& \sup _{P \in \mathcal{P}} P\left(A_n\right) \\
& \leq \sup _{P \in \mathcal{P}} P\left\{\left|  \Delta_n +\frac{F_n^2}{\hat{\sigma}^2_{+,U} + \hat{\sigma}^2_{-,U}} - U - \frac{F^2}{\sigma^2_{+,U} + \sigma^2_{-,U}}  \right|
> 
\frac{1}{2}\left(U+\frac{F^2}{\sigma^2_{+,U} + \sigma^2_{-,U}}\right)\right\} \\
& \quad \rightarrow 0.
\end{aligned}
$$
Denote by $c_\alpha$ the $(1-\alpha)$-th quantile of the $\chi_{1}^2$ distribution, then the limiting power function can be written as,
$$
\begin{aligned}
P\left(R_{n, \alpha}\right) & =P\left(\tilde{U}_n + \Delta_n +\frac{F_n^2}{\hat{\sigma}^2_{+,U} + \hat{\sigma}^2_{-,U}}>\frac{c_\alpha}{nh}\right) \\
& \geq P\left\{\tilde{U}_n +\frac12 \left( U+\frac{F^2}{\sigma^2_{+,U} + \sigma^2_{-,U}}\right)>\frac{c_\alpha}{nh}, A_n^C\right\} \\
& \geq P\left\{nh \tilde{U}_n > c_\alpha-\frac{nh}{2}\left(U+\frac{F^2}{\sigma^2_{+,U} + \sigma^2_{-,U}}\right)\right\}-P\left(A_n\right) \\
& \geq P\left\{nh \tilde{U}_n > c_\alpha-\frac{nh}{2}\left(U+\frac{F^2}{\sigma^2_{+,U} + \sigma^2_{-,U}}\right)\right\}-\sup_{P \in \mathcal{P}} P\left(A_n\right).
\end{aligned}
$$
This implies that, for the sequence of hypotheses $\left\{H^{(n)}_A\right\}$ defined in Section \ref{sec:consistency},
$$
\begin{aligned}
\lim _{n \rightarrow \infty} \beta_{H_A^{(n)}} & =\lim _{n \rightarrow \infty}\left\{\inf _{H_A^{(n)}} P\left(R_{n, \alpha}\right)\right\} \\
& \geq \lim _{n \rightarrow \infty} P\left\{nh \tilde{U}_n>c_\alpha-\frac{nh}{2}\left(b_n+\frac{a_n^2}{\sigma^2_{+,U} + \sigma^2_{-,U}}\right)\right\}-\lim _{n \rightarrow \infty}\left\{\sup _{P \in \mathcal{P}} P\left(A_n\right)\right\} \\
& =\lim _{n \rightarrow \infty} P\left\{ nh \tilde{U}_n>c_\alpha-\frac{nh}{2}\left(b_n+\frac{a_n^2}{\sigma^2_{+,U} + \sigma^2_{-,U}}\right)\right\} .
\end{aligned}
$$
Then, since $nh \tilde{U}_n \to \chi_{1}^2$ in distribution, one obtains that if $a_n$ is such that $(nh)^{1 / 2} a_n \rightarrow \infty$ or if $b_n$ is such that $nh b_n \rightarrow \infty$, then $\lim _{n \rightarrow \infty} \beta_{H_n}=1$.
\end{proof}

\paragraph{Proof of Lemma \ref{lemma:consistency_sup_p}.} 
\begin{proof}
I continue to suppress the $\oplus$ subscripts throughout the proof to ease notation. 
\begin{list}{\textit{Proof of Lemma \ref{lemma:consistency_sup_p}(\alph{enumi}).}}{\usecounter{enumi} \setlength{\leftmargin}{0pt} \setlength{\itemindent}{\labelwidth} \addtolength{\itemindent}{\labelsep}}
\item  We have,
$$
\begin{aligned}
& P(|\hat{V}_+ - V_+| > \epsilon) \\
&= P\bigg[\bigg| \frac{1}{n} \sum_{i=1}^n s_{+,in}(h) \big( d^2(\hat{l}_+, Y_i) - d^2(m_+, Y_i) \big) \\
&\quad + \frac{1}{n} \sum_{i=1}^n s_{+,in}(h) d^2(m_+, Y_i) - V_+ \bigg| > \epsilon \bigg] \\
&\leq C_n + D_n,
\end{aligned}
$$
where
\begin{align*}
C_n & =P\left\{\left|\frac{1}{n} \sum_{i=1}^n s_{+,in}(h) \big( d^2(\hat{l}_+, Y_i) - d^2(m_+, Y_i) \big)\right|>\epsilon / 2\right\}, \\
D_n & =P\left[\left|\frac1n \sum_{i=1}^n s_{+,in}(h) d^2(m_+, Y_i) - V_+\right|>\epsilon / 2\right] .
\end{align*}
Since $\inf_{\omega \in \Omega} \left( \lim_{x \to c^+} E\left[d^2(\omega, Y) - d^2(m_{+}, Y) \mid X=x\right] \right) =0$, we have
$$
\begin{aligned}
& \left|\frac{1}{n} \sum_{i=1}^n s_{+,in}(h) \left( d^2\left(\hat{l}_+, Y_i\right)- d^2\left(m_+, Y_i\right) \right) \right| \\
& =\left|\inf _{\omega \in \Omega}\left\{\frac{1}{n} \sum_{i=1}^n s_{+,in}(h) \left( d^2\left(\omega, Y_i\right)-d^2\left(m_+, Y_i\right)\right)\right\}\right| \\
 & \leq \sup _{\omega \in \Omega}\left|\frac{1}{n} \sum_{i=1}^n s_{+,in}(h) \left( d^2\left(\omega, Y_i\right)-d^2\left(m_+, Y_i\right)\right) - \left(\lim_{x \to c^+}E\left[d^2(\omega, Y)\mid X=x\right] - V_+\right) \right| \\
& \coloneqq \sup _{\omega \in \Omega}\left|M_n(\omega)\right|,
\end{aligned}
$$
analogous to the proof of Lemma \ref{lemma:clt_first_term}.

Further mimicking the steps in that proof, define, 
\[
B_{m_+, R} \coloneqq \left\{\sup_{\omega \in \Omega}\left|\frac{1}{n} \sum_{i=1}^n\left[s_{+, i n}(h)-s_{+, i}(h)\right] D_{m_\oplus, i}(\omega)\right| \leq R   \mathrm{diam}^2(\Omega) (n h)^{-1 / 2}\right\}
\]
noting that I replaced $\delta$ with $\mathrm{diam}(\Omega)$ using the boundedness of the metric space to bound the supremum over $\omega$, as opposed to using $d(\omega, m_{+,\oplus}) < \delta$ in Lemma \ref{lemma:clt_first_term}. Similarly, adapting the function class in Lemma \ref{lemma:l_hat-l_tilde} to, 
\[
\mathcal{M}_{m_+, n \delta}=\left\{g_\omega-g_{m_{+, \oplus}}: \omega \in \Omega\right\},
\]
an envelope function for it is,
\[
G_{m_+, n \delta}(z)=\frac{4 \operatorname{diam}^2(\Omega) }{\sigma_{+, 0}^2} 1(z \geq c) K_h(z-c)\left|\mu_{+, 2}-\mu_{+, 1}(z-c)\right|,
\]
where $E[G_{m_+, n \delta}(X)] = O(h^{-1})$. Again using Theorems 2.7.11 and 2.14.2 of \citet{vaart1996weak} together with the finite entropy integral in \ref{asspT:entropy}, 
\[
E\left(\sup_{\omega \in \Omega}\left|\frac{1}{n} \sumn s_{+, i}(h) D_{m_{\oplus},i}(\omega, x)-E\left[s_{+, i}(h) D_{m_\oplus,i}(\omega, x)\right]\right|\right)=O\left((n h)^{-1 / 2}\right).
\]
Combining, we obtain $E\left( \mathrm{1}_{B_{m_{\oplus},R}} \sup _{\omega \in \Omega}\left|M_n(\omega)\right|\right) \leq a (nh)^{-1 / 2}$, for some $a>0$. Further, using the definition of $B_n(\omega)$ in Lemma \ref{lemma:clt_first_term}, we also have by the results in Lemma \ref{lemma:fan_gijbels},
\[
\sup_{\omega \in \Omega} |B_n(\omega)| \leq b 2\mathrm{diam}^2(\Omega) h^2
\]
for some $b>0$. 
Combining these three bounds, we get for some $b>0$, by Markov's inequality, 
\[
C_n \leq 2 a \epsilon^{-1} (nh)^{-1/2} + 2 b \epsilon h^2 + P(B_{m_{\oplus}, R}).
\]
To bound $D_n$, note that,
\begin{align*}
& \left|\frac{1}{n} \sum_{i=1}^n s_{+,in}(h) d^2\left(m_+, Y_i\right)-V_+\right| \\
& \leq \sup _{\omega \in \Omega}\left|\frac{1}{n} \sum_{i=1}^n s_{+,in}(h) d^2\left(\omega, Y_i\right)-\lim_{x \to c^+}E\left[d^2(\omega, Y) \mid X=x\right]\right| \coloneqq \sup _{\omega \in \Omega}\left|H_n(\omega)\right|.
\end{align*}
By similar arguments, but using $d^2(Y_i, \omega)$ instead of $D_{m_{\oplus},i}(\omega)$ to analogously define $B_{\omega,R}$, one obtains that
$$
E\left(\sup _{\omega \in \Omega}\left|H_n(\omega)\right|\right) \leq d (nh)^{-1 / 2} + e h^{-2}  + P(B_{\omega, R}),
$$
for some $d,e>0$ which again depends on the finite entropy integral $J$ and an analogous envelope function that bounds a function of $d^2(Y_i, \omega)$. By Markov's inequality, 
\[
D_n \leq \epsilon^{-1}  (d (nh)^{-1 / 2} + e h^2) + P(B_{\omega, R}).
\]
Combining, we obtain that,
\[
\sup_{P \in \mathcal{P}} P(|\hat{V}_+ - V_+| > \epsilon) \leq \epsilon^{-1} \left( (a+c) (nh)^{-1/2} + (b+d)h^2  \right) +   P(B_{m_{\oplus}, R}) +  P(B_{\omega, R}) \to 0.
\]

\item Proving the uniform consistency of $\hat{\sigma}^2$ only requires proving it for $\frac1n \sumn s_{+,in} d^2(\hat{l}_+, Y_i)$ and its target $\lim_{x \to c^+}E\left[ d^4(m_+, Y) \mid X=x \right]$ since the rest follows from (a) by continuity. As in the proof of (a), we have,
\begin{align*}
& P\left[\left|\frac{1}{n} \sum_{i=1}^n s_{+,in}(h) d^4\left(\hat{l}_+, Y_i\right)-\lim_{x \to c^+}E\left[d^4(m_+, Y) \mid X=x\right]\right|>\epsilon\right] \\
= & P\left[ \left|\frac{1}{n} \sum_{i=1}^n  s_{+,in}(h) \left( d^4\left(\hat{l}_+, Y_i\right)- d^4\left(m_+, Y_i\right)\right) \right. \right. \\
& \left. \left. +\frac{1}{n} \sum_{i=1}^n s_{+,in}(h) d^4\left(m_+, Y_i\right)-\lim_{x \to c^+}E\left[d^4(m_+, Y) \mid X=x\right]\right|>\epsilon\right] \\
\leq & A_n+B_n
\end{align*}
where
\begin{align*}
& A_n=P\left\{\left|\frac{1}{n} \sum_{i=1}^n s_{+,in}(h) \left( d^4\left(\hat{l}_+, Y_i\right)- d^4\left(m_+, Y_i\right) \right)\right|>\epsilon / 2\right\}, \\
& B_n=P\left[\left| \frac{1}{n} \sum_{i=1}^n s_{+,in}(h) d^4\left(m_+, Y_i\right)-\lim_{x \to c^+}E\left[d^4(m_+, Y)\mid X=x\right]\right|>\epsilon / 2\right].
\end{align*}
Since we have,
\begin{align*}
& \left|\frac{1}{n} \sum_{i=1}^n s_{+,in}(h) \left( d^4\left(\hat{l}_+, Y_i\right)- d^4\left(m_+, Y_i\right) \right) \right|  \\ 
& \leq 2 \operatorname{diam}^2(\Omega) \left|\frac{1}{n} \sum_{i=1}^n s_{+,in}(h) \left( d^2\left(\hat{l}_+, Y_i\right)-d^2\left(m_+, Y_i\right) \right) \right|,
\end{align*}
the proof of (a) implies,
\[
A_n \leq  O_p((nh)^{-1/2} + h^2) + o(1).
\]
Furthermore, 
\begin{align*}
& \left|\frac{1}{n} \sum_{i=1}^n s_{+,in}(h) d^4\left(m_+, Y_i\right)-\lim_{x \to c^+}E\left[d^4(m_+, Y) \mid X=x\right]\right| \\ 
& \leq \sup _{\omega \in \Omega}\left|\frac{1}{n} \sum_{i=1}^n s_{+,in}(h) d^4\left(\omega, Y_i\right)-\lim_{x \to c^+}E\left[d^4(\omega, Y) \mid X =x\right]\right|\coloneqq \sup _{\omega \in \Omega}\left|K_n(\omega)\right|.
\end{align*}
Similar arguments as in (a), but using $d^4(Y_i, \omega)$ instead of $D_{m_{\oplus},i}(\omega)$, imply for $c>0$,
\[ 
E\left[ \sup_{\omega \in \Omega} |K_n(\omega) | \right] \leq  O_p((nh)^{-1/2} + h^2) + o(1).
\]
Combining, we obtain,
\[
\sup_{P \in \mathcal{P}} P\left[\left|\frac{1}{n} \sum_{i=1}^n s_{+,in}(h) d^4\left(\hat{l}_+, Y_i\right)-\lim_{x \to c^+}E\left[d^4(m_+, Y) \mid X=x\right]\right|>\epsilon\right] \to 0. \quad \qedsymbol
\]

\item Observe that 
$$
\begin{aligned}
P(|\hat{V}_p - V_p| > \epsilon) 
&= P\bigg[\bigg| \frac{1}{n} \sum_{i=1}^n s_{p,in}(h) \big( d^2(\hat{l}_p, Y_i) - d^2(m_p, Y_i) \big) \\
&\quad + \frac{1}{n} \sum_{i=1}^n s_{p,in}(h) d^2(m_p, Y_i) - V_p \bigg| > \epsilon \bigg] \\
&\leq C_{p,n} + D_{p,n},
\end{aligned}
$$
where
\begin{align*}
C_{p,n} & =P\left\{\left|\frac{1}{n} \sum_{i=1}^n s_{p,in}(h) \big( d^2(\hat{l}_p, Y_i) - d^2(m_p, Y_i) \big)\right|>\epsilon / 2\right\}, \\
D_{p,n} & =P\left[\left|\frac1n \sum_{i=1}^n s_{p,in}(h) d^2(m_p, Y_i) - V_p\right|>\epsilon / 2\right].
\end{align*}
Since $\inf_{\omega \in \Omega} \left(M_p(\omega) - V_p\right)=0$, we have
$$
\begin{aligned}
& \left|\frac{1}{n} \sum_{i=1}^n s_{p,in}(h) \left( d^2\left(\hat{l}_p, Y_i\right)- d^2\left(m_p, Y_i\right) \right) \right| \\
& =\left|\inf _{\omega \in \Omega}\left\{\frac{1}{n} \sum_{i=1}^n s_{p,in}(h) \left( d^2\left(\omega, Y_i\right)-d^2\left(m_p, Y_i\right)\right)\right\}\right| \\
 & \leq \sup _{\omega \in \Omega}\left|\frac{1}{n} \sum_{i=1}^n s_{p,in}(h) \left( d^2\left(\omega, Y_i\right)-d^2\left(m_p, Y_i\right)\right)-\left(M_p(\omega) - V_p\right) \right| \\
& \coloneqq \sup _{\omega \in \Omega}\left|M_{p,n}(\omega)\right|.
\end{aligned}
$$
Then the rest of the proof follows from the same proof strategy as in (a), but using the bounds established in Proposition \ref{prop:pooled_consistency}. \qedhere
\end{list}
\end{proof}

\section{Additional Results} \label{app:add_results}

\subsection{Tables}

\begin{table}[htbp!]
\scriptsize
\input{tabs/rdd_balance_detailed}
\caption{Non-Compete Application: Scalar-Valued RD Estimates}
\label{tab:rdd}
\floatfoot{\scriptsize \textit{Note:} Each row reports the RD estimate for the discontinuity in the specified covariate at the annual
  income threshold above which non-compete agreements become enforceable, using the local linear estimator of \citet{calonico2014robust}
  with MSE-optimal bandwidth selection. Robust standard errors in parentheses. \\
  Significance levels: * \(p<0.10\), ** \(p<0.05\), *** \(p<0.01\).}
\end{table}

\begin{table}[ht]
  \centering
  \footnotesize
  \caption{Aggregation of the 26 EORA Sectors into Eight Categories}
  \label{tab:eora_agg_map_short}
  \begin{tabular}{@{}lp{10cm}@{}}
    \toprule
    \textbf{Aggregate Sector} & \textbf{Constituent EORA Sectors} \\
    \midrule
    Agri.\ \& Fish.       & Agriculture; Fishing \\
    Mining                & Mining and Quarrying \\
    Manufacturing         & Food \& Beverages; Textiles and Wearing Apparel; Wood and Paper; Petroleum, Chemical and Non-Metallic Mineral Products; Metal Products; Electrical and Machinery; Transport Equipment; Other Manufacturing; Recycling \\
    Utilities             & Electricity, Gas and Water \\
    Construction          & Construction; Maintenance and Repair \\
    Trade \& Hosp.        & Wholesale Trade; Retail Trade; Hotels and Restaurants; Transport; Post and Telecommunications; Re-export \& Re-import \\
    Finance \& Business   & Financial Intermediation and Business Activities \\
    Public \& Other       & Public Administration; Education, Health and Other Services; Private Households; Others \\
    \bottomrule
  \end{tabular}
\end{table}

\begin{table}
\footnotesize
\caption{I-O Network Application: Scalar-Valued RD Estimates}
\input{tabs/lag2_scalar_rdd}
\label{tab:io_scalar_rd}
\floatfoot{\scriptsize \textit{Note}: each row reports the regression discontinuity estimate for the jump in the specified scalar outcome at the GNI per capita cutoff which triggers GSP graduation, using the local linear estimator of \citet{calonico2014robust} with MSE-optimal bandwidth. Robust bias‐corrected standard errors are shown in parentheses.

 ``Centrality'' is the average column‐sum of the technical‐coefficient matrix $A$, capturing how much intermediate input the economy supplies per unit of output; ``Economic Complexity'' is the average number of inputs per sector whose share in $A$ exceeds 1\%, a proxy for the diversity of domestic production; ``Manufacturing Intensity'' and ``Services Intensity'' are the mean diagonal entries of $A$ for the manufacturing and services aggregates, respectively, measuring each sector’s reliance on its own output; and ``Multiplier Effect'' is the average column‐sum of the Leontief inverse ($(I - A)^{-1}$), which measures total (direct + indirect) output generated by a one‐unit increase in final demand.}
\end{table}

\begin{table}
\small
\input{tabs/leontief_top_changes_rv_2015_io_lag_2}
\label{tab:linkage_ranks}
\end{table}

\end{supplement}

%% file: tabs/simulation_results_tabular.tex
{ % Start a group to keep the change local
\renewcommand{\arraystretch}{0.8} % Adjust this value (e.g., 0.75) for more/less compression
\begin{tabular}{lccr}
  \toprule
  Metric Space & N & Size & Power \\ 
  \midrule
  \multirow{3}{*}{Covariance} & 200 & 0.031 & 0.997 \\ 
   & 500 & 0.046 & 0.997 \\ 
   & 1000 & 0.048 & 1.000 \\ 
  \cmidrule(lr){2-4}
  \multirow{3}{*}{Density} & 200 & 0.038 & 0.999 \\ 
   & 500 & 0.043 & 1.000 \\ 
   & 1000 & 0.043 & 1.000 \\ 
  \cmidrule(lr){2-4}
  \multirow{3}{*}{Network} & 200 & 0.021 & 1.000 \\ 
   & 500 & 0.062 & 1.000 \\ 
   & 1000 & 0.055 & 1.000 \\ 
  \bottomrule
\end{tabular}
} % End the group

%% file: tabs/rdd_balance_detailed.tex
\begin{tabular}{lc}
\toprule
Covariate & Coefficient \\
 & (Std. Error) \\
\midrule
Age & 1.87 \\
 & (6.46) \\
[0.5ex]
Female & 0.06 \\
 & (0.243) \\
[0.5ex]
Education Level & -0.082 \\
 & (0.95) \\
[0.5ex]
Work in Different State & -0.046 \\
 & (0.099) \\
[0.5ex]
WFH Any & -0.129 \\
 & (0.229) \\
[0.5ex]
WFH Full-Time & 0.088 \\
 & (0.214) \\
[0.5ex]
WFH 0 days/week & 0.101 \\
 & (0.224) \\
[0.5ex]
WFH 1 day/week & -0.036 \\
 & (0.042) \\
[0.5ex]
WFH 2 days/week & -0.034 \\
 & (0.042) \\
[0.5ex]
WFH 3 days/week & 0 \\
 & (0) \\
[0.5ex]
WFH 4 days/week & -0.165 \\
 & (0.114) \\
[0.5ex]
WFH 5 days/week & 0.112 \\
 & (0.202) \\
\midrule
Observations & 569 \\
\bottomrule
\end{tabular}

%% file: tabs/lag2_scalar_rdd.tex
\begin{tabular}[t]{lr}
\toprule
Outcome & Coefficient \\
\midrule
Centrality & -0.0331 \\
  & (0.1075) \\
Economic Complexity & 0.2702 \\
  & (0.4452) \\
Manufacturing Intensity & -0.0181 \\
  & (0.1168) \\
Multiplier Effect & -0.4592 \\
  & (1.4789) \\
Services Intensity & -0.0091 \\
  & (0.0187) \\
\midrule
\multicolumn{1}{l}{Observations} & \multicolumn{1}{r}{140} \\
\bottomrule
\end{tabular}

%% file: tabs/leontief_top_changes_rv_2015_io_lag_2.tex
\centering
\caption{Preferential Tariff Loss: Industry Pairs with Largest Changes}
\begin{subtable}{0.45\textwidth}
\centering
 \begin{center}
\caption{Top 5 Increases}
\centering
\begin{tabular}[t]{rllr}
\toprule
Rank & Sector A & Sector B & Change\\
\midrule
1 & Construction & Public \& Other & 0.108\\
2 & Mining & Construction & 0.062\\
3 & Mining & Utilities & 0.039\\
4 & Utilities & Construction & 0.034\\
5 & Utilities & Public \& Other & 0.034\\
\bottomrule
\end{tabular}
\end{center} 
\end{subtable}\hfill
\begin{subtable}{0.45\textwidth}
  \centering
 \begin{center}

\caption{Top 5 Decreases}
\centering
\begin{tabular}[t]{rllr}
\toprule
Rank & Sector A & Sector B & Change\\
\midrule
1 & Mining & Manufacturing & -0.468\\
2 & Mining & Trade \& Hosp. & -0.361\\
3 & Manufacturing & Construction & -0.257\\
4 & Manufacturing & Trade \& Hosp. & -0.209\\
5 & Construction & Trade \& Hosp. & -0.198\\
\bottomrule
\end{tabular}
\end{center} 
\end{subtable}